\documentclass[envcountsect,envcountsame,runningheads]{llncs}
\usepackage[T1]{fontenc}
\usepackage[utf8]{inputenc}
\usepackage{hyperref}

\usepackage{ksano_logic}

\makeatletter
\@ifpackageloaded{fixme}{}{\usepackage[final]{fixme}} %inline
\makeatother
\FXRegisterAuthor{ks}{aks}{KS}%Katsuhiko
\FXRegisterAuthor{tmil}{atmil}{TL minor}%Tadeusz
\FXRegisterAuthor{tl}{atl}{TL major}%Tadeusz
\definecolor{mydarkgreen}{rgb}{0,0.34,0}
\newcommand{\tlnt}[1]{\tmilnote[inline,marginclue]{\textcolor{mydarkgreen}{#1}}}

\usepackage{ifthen}

\makeatletter
\newif\if@full
 \@fulltrue    %%%%%%%%%%% <--- switch between \@fullfalse and \@fulltrue to show appendices etc!
\relax
\let\iffull\if@full

\iffull
\usepackage{geometry}

\else\fi

\usepackage{multirow}
\usepackage{tabularx}
\newcolumntype{L}[1]{>{\raggedright\let\newline\\\arraybackslash\hspace{0pt}}m{#1}}
\newcolumntype{C}[1]{>{\centering\let\newline\\\arraybackslash\hspace{0pt}}m{#1}}
\newcolumntype{R}[1]{>{\raggedleft\let\newline\\\arraybackslash\hspace{0pt}}m{#1}}

\newcommand{\powset}{\mathcal{P}}
\newcommand{\states}[1]{\powset^{\supseteq}%_{\neq \varnothing}
(#1)}

\newcommand{\lna}[1]{\ensuremath{\mathsf{#1}}}
\newcommand{\inqbqa}[1]{\lna{InqBQ}_{#1}}
\newcommand{\inqbq}{\lna{InqBQ}}
\newcommand{\inqbqn}{\inqbqa{n}}
\newcommand{\inqbqo}{\inqbqa{<\omega}}

\newcommand{\IntModels}[1]{\ensuremath{\lna{IntModels}(#1)}}
\newcommand{\InjModels}[1]{\ensuremath{\lna{ReflInfModels}(#1)}}

\newcommand{\fwdm}{\states{W},\lft{\mM}}

\newcommand{\ndm}{\mathfrak{N}, \mathfrak{J}}
\newcommand{\nng}{\ndm, n \Vdash_g}
\newcommand{\nog}{\ndm, o \Vdash_g}

\newcommand{\mM}{\ensuremath{\mathfrak{M}}}
\newcommand{\mN}{\ensuremath{\mathfrak{N}}}
\newcommand{\mI}{\ensuremath{\mathfrak{I}}}

\newcommand{\lft}[1]{\hat{#1}}

\renewcommand{\models}{\Vdash}

\renewcommand{\phi}{\varphi}

\newcommand{\lang}{\mathsf{L}}
\newcommand{\langrex}{\mathsf{L_{rex}}}

\title{Bounded Inquisitive Logics: \\
Sequent Calculi and Schematic Validity}

\author{Tadeusz Litak\inst{1}
\orcidID{0000-0003-2240-3161} 
\and
Katsuhiko Sano\inst{2}
\orcidID{0000-0002-7780-423X} 
}

\institute{\iffull
 Dipartimento di Ingegneria Elettrica e delle Tecnologie dell'Informazione, \\ \else\fi
 University of Naples Federico II\\ \email{tadeusz.litak@gmail.com}  \and  
Faculty of Humanities and Human Sciences, Hokkaido University \\ \email{v-sano@let.hokudai.ac.jp}
}

\begin{document}

\maketitle

\begin{abstract}
Propositional inquisitive logic is the limit of its $n$-bounded approximations. In the predicate setting, however, this does not hold anymore, as discovered by Ciardelli and Grilletti \cite{Ciardelli2022a}, who also found complete axiomatizations of $n$-bounded inquisitive logics $\inqbq_{n}$, for every fixed $n$. We introduce cut-free labelled sequent calculi for these logics. %$n$-bounded inquisitive predicate logics, in which the cut rule is admissile. 
%These calculi follow the ideas of Sano 2011 (which investigated $\inqbq_{2}$), and make more intensive use of formulas that are \textit{schematically valid}  in the bounded case than the original proposal by Ciardelli and Grilletti. 
We illustrate the intricacies of \textit{schematic validity} in such systems by showing that the well-known Casari formula is \textit{atomically} valid in (a weak sublogic of) predicate inquisitive logic $\mathsf{InqBQ}$, fails to be schematically valid in it, and yet is schematically valid under the finite boundedness assumption. The derivations in our calculi, however, are guaranteed to be schematically valid whenever a single specific rule is not used. 
%We demonstrate the deductive power and flexibility of our calculi by providing uniform derivations of formulas schematically valid in each  $\inqbq_{n}$, including the Casari scheme, but also the constant domain scheme, Kuroda or variants of Kreisel-Putnam, as well as derivations of formulas that are only atomically valid (e.g., the double negation law). Our Gentzen-style formalism appears optimal in several aspects.  
\end{abstract}

\

%\end{frontmatter}

\iffull
\begin{footnotesize}
\paragraph*{Disclaimer}
This is a modified and expanded version of a paper accepted for TABLEAUX 2025. In particular, readers should note that the numeration of environments is different in the conference version. 
\end{footnotesize}
\else\fi

\section{Introduction}

Inquisitive logic~\cite{Ciardelli2018,Ciardelli2016,Ciardelli2017} provides a framework for studying both declarative and interrogative sentences in one setting. On the propositional level, it enriches classical logic with {\em inquisitive disjunction} $\inqd$, and on the predicate level, it also adds the \emph{inquisitive existential quantifier} $\inqe$. Assuming $S(x,y)$ is the predicate ``$x$ sings for $y$'', $e$ is the constant denoting Eric Clapton and $g$ denotes Gottlob Frege, we write ``Does Eric Clapton sing for Gottlob Frege?'' as $?S(e,g) := S(e,g) \inqd \neg S(e,g)$. The question ``Who is some person that Eric Clapton sings for?'' translates to $\Exi{x}S(e,x)$, whereas ``Which are the people Eric Clapton sings for?'' is rendered as $\Any{x}?S(e,x)$, and ``Is Eric Clapton singing for everybody?'' gets expressed as $?\Any{x}S(e,x)$.

%\tlnt{Referee 3 insists that inquisitive disjunction/quantifier are prefixed by the definite article in the above paragraph and a few other places (like beginning of Section 3). Not sure if it is obligatory or even correct in fact}

%Semantically, a formula is evaluated not by a single state but by a set of states called a {\em team}. This semantic feature is also a core of (propositional) {\em dependence logic}~\cite{fanthesis}, where we can study the notion of {\em functional dependence} $\mathrm{dep}(q;p)$, ``$q$ truth-functionally determines $p$''.  In this sense, the semantics for dependence logic is called {\em team} semantics. Functional dependency $\mathrm{dep}(q;p)$ can be understood as an implication from the question $?q$ to the question $?p$ (see \cite{Ciardelli2016} for more detail).
%The ideas of inquisitive logic and dependence logic have been generalized  to  various  non-classical propositional bases, i.e., modal logic~\cite{va08,EbbingL12,EHMMVV13}, (dynamic) epistemic logic~\cite{Ciardelli2015}, intuitionistic logic~\cite{Puncochar2016,Puncochar2017,Ciardelli2020}, substructural logic~\cite{Puncochar2019}, etc.

%\tlnt{need a discussion of existential quantifier. Also, do we need to discuss team semantics here, in fact?}

Semantically, \emph{relational information models} reflect different (combinations of) potential answers to questions as \emph{states}, which are just collections of FO models (structures) over some fixed domain of individuals. A singleton state is maximal under $\supseteq$, settling one way or another all questions. The interpretation of an atomic predicate at a state $s$ is obtained as the intersection of its denotations in members of $s$. Thus, inquisitive logic can be seen not only as an expansion of classical logic with additional operators, but also as a formalism extending (as a set of theorems) the intuitionistic logic of constant domains \lna{CD} \cite{gornemann71}. In other words, $\inqd$ and $\inqe$ are just standard intuitionistic connectives, whereas $\vee$ and $\exists$ are simply de Morgan duals of $\wedge$ and $\forall$, respectively. Atomic formulas, however, are evaluated as \emph{regular} upsets satisfying the double negation law. Anyone familiar with, e.g., the G\"odel-Gentzen translation \cite[Ch. 2.3]{TroelstraVD},  \cite[Ch. 6--7]{SorensenU06:book}, \cite[Ch. 2]{ChagrovZ97:ml} \cite{FerreiraO12:full,litakpr17} can see that formulas involving $\inqd$ and $\inqe$ are the only ones that can behave in a non-classical way under this restriction.%, which justifies the above perspective on inquisitive logic as extension of classical logic with these two additional connectives.

The failure of some non-atomic instances of Boolean laws implies that unlike standard superintuitionistic logics, inquisitive logic is not closed under schematic substitutions (see Section \ref{sec:casari}). This naturally leads to the question of \emph{schematic validity} in inquisitive logic: what is its \emph{schematic core/fragment}, i.e., the largest standard superintuitionistic logic contained in it?  For the propositional inquisitive logic \lna{InqL}, Ciardelli \cite{Ciardelli09}  established that its schematic fragment is exactly Medvedev's logic \lna{ML} of finite problems or finite (topless) boolean cubes. Conversely, \lna{InqL} can be obtained as the negative counterpart of \lna{ML}, i.e., the collection of formulas whose negatively substituted variants (replacing each atom with its negation) belong to \lna{ML}. To the best of our knowledge, the corresponding first-order question has not been addressed.
 %---difficult in the light of essentially second-order character of schematic validity and complications involving substitutions for predicate symbols---has not been addressed.
 
Another set of questions is brought by the idea of restricting the size of the models. Of course, by Trakhtenbrot's theorem, the restriction to (unbounded) finite domains of individuals would block recursive axiomatizability in the presence of at least one binary relation symbol. Moreover, as noted by Ono \cite{OnoH:1973}, there are axiom \textit{schemes} in the sense of Section \ref{sec:casari} that can sense finiteness of the domain of individuals. %; consequently, substitution-closed superintuitionistic predicate logics are not necessarily \emph{intermediate} in the sense of being contained in the classical predicate calculus. 
However, in the inquisitive setting, we can also restrict the cardinality of the \emph{collection of structures}, and consequently of \emph{available states}; note that finite boundedness defined this way is a generalization of the finite model property defined in terms of individuals. In the propositional case, this issue also has been investigated by Ciardelli \cite{Ciardelli09}: not only can all inquisitive non-theorems be refuted with a countable collection of possible worlds/structures, but \lna{InqL} is in fact the limit of its $n$-bounded approximations. Furthermore, \emph{op.cit.} proves that each such $n$-bounded approximation is obtained as the negative counterpart of \lna{ML} extended with the axiom limiting the depth of frames to $n$.

The second named author~\cite{Sano2011c} asked whether analogues of boundedness results discussed in the preceding paragraph hold in the predicate case, and provided two complete calculi (Hilbert-style and tree-sequent) for the two-bounded system $\inqbq_{2}$. Recently, Ciardelli and Grilletti \cite{Ciardelli2022a} have provided complete axiomatizations of $\inqbq_{n}$ for each $n \in \omega$, and proved that $\inqbq \subsetneq \bigcap\inqbq_{n}$. %unlike the propositional case, predicate inquisitive logic is not the limit of finitely bounded ones.
Disregarding their specific counterexample, one can see this strict inclusion as follows. 
%Experts in correspondence theory are well-aware that 
In the propositional case, \emph{there are no intuitionistic analogues of the Grzegorczyk and L\"ob formulas, which expresses the presence of infinite ascending chains} \cite[p. 175]{ChagrovZ97:ml}. On the other hand, there are numerous such formulas are known in the predicate case, both for ascending and descending chains \cite{Casari83,Esakia98:bsl,MinariEA90}, with the Casari formula \cite{Casari83} being one of the most salient examples (cf. Theorem \ref{th:casari} for its semantic characterization). Such formulas must be schematically valid in $\inqbq_{n}$, for any $n \in \omega$. But it does not immediately follow that they fail in $\inqbq$: the restriction of atoms to regular upsets makes finding relational information countermodels challenging. The negative substitution of the Casari axiom is in fact an intutionistic theorem (Theorem \ref{th:casat}). In the unbounded case there exists a refuting instance of the Casari scheme and a suitable relational information model, but the proof is far from trivial (Theorem \ref{th:casfail}).

Given the difficulty of providing such semantic proofs, one can wonder if there are convenient proof calculi to work with. Several proof systems have been proposed for propositional inquisitive logic\footnote{A Hilbert system and a natural deduction system for the propositional inquisitive logic $\mathsf{InqB}$ are presented in~\cite{Ciardelli2010} and~\cite{Ciardelli2022}, respectively. A display calculus for $\mathsf{InqB}$ is also provided in~\cite{Frittella2016}. Labelled G3-style sequent calculi for $\mathsf{InqB}$ were proposed in~\cite{Chen2017,Muller2023}; see Footnote \ref{ft:propcalc}.}. %where relational atoms of the form ``$s\supseteq t$' are used to capture the inclusion relation between states $s$ and $t$.}.  
In the predicate setting, proof calculi have been developed for certain syntactic fragments or for logics tailored to specific classes of models. Grilletti~\cite{Grilletti2021} presents a sound and complete natural deduction system for the classical antecedent fragment of first-order inquisitive logic. For $\mathsf{InqBQ}_n$, Ciardelli and Grilletti~\cite{Ciardelli2022a} provide a sound and strongly complete natural deduction calulus, %However, the proof-theoretic study of first-order inquisitive logic in terms of sequent calculi remains limited. Sano~\cite{Sano2011c} proposed a sound and complete labelled sequent calculus for $\mathsf{InqBQ}_2$, but this system is not G3-style and provides only semantic cut elimination. Ciardelli and Grilletti \cite{Ciardelli2022a} provide heavily signature-dependent axiomatizations of finitely bounded fragments, 
 which however does not appear promising in the context of investigating schematic validity in bounded systems, and in several other typical proof-theoretic applications such as proof search (cf. Remark~\ref{rem:comparison} below). We instead leverage the labelled sequent system for $\inqbq_{2}$ provided by Sano~\cite{Sano2011c}, generalizing it to a semantically complete, cut-free G3-style calculus for $\inqbq_{n}$ for any $n \in \omega$. As far as we are aware, this is the first G3-style calculus for such systems. %finitely bounded predicate inquisitive logics. 
 In our sequent calculus, all inference rules are height-preserving invertible, weakening and contraction are height-preserving admissible, and the cut rule is admissible via a syntactic argument (Sano~\cite{Sano2011c} only shows semantic cut elimination for the  calculus for $\inqbq_{2}$ considered therein).  Furthermore, the derivations in our calculi are guaranteed to be schematically valid whenever a single specific rule is not used (Theorem~\ref{thm:generation_schematic_validities}).

We proceed as follows: Section \ref{sec:prelim} sets up the syntax of first-order inquisitive logic and its Kripke semantics. Section \ref{sec:casari} discusses the issue of validity of the Casari scheme. %observes that while the Casari scheme is valid in $\inqbqo$, some of its non-atomic instances fail in $\inqbq$. 
After introducing our labelled sequent calculus for the finitely bounded first-order inquisitive logics in Section \ref{sec:labseq}, Section \ref{sec:semcompl} proves the semantic completeness of the calculi, and Section \ref{sec:cutelim} establishes the admissibility of cut using a syntactic argument. Tying together the two main threads of the paper, Proposition \ref{prop:dercasari} illustrates how the Casari scheme is derived using structural rules when a specific finite bound is fixed; more such derivations are provided in the accompanying material. Section \ref{sec:conclusions} concludes the paper with further directions of research.

 \if0
Apart from key metalogical results, i.e., completeness (Section \ref{sec:semcompl}) and admissibility of cut and other standard rules (Section \ref{sec:cutelim}), we provide derivations of schemes valid in each  $\inqbq_{n}$, including the Casari scheme (Proposition \ref{prop:dercasari}, but also \lna{CD} (Example \ref{ex:cd}), the Kuroda scheme (Proposition \ref{prop:kurocalc}) or variants of Kreisel-Putnam (Propositions \ref{prop:KP} and \ref{prop:EKP}). % as well as derivations of formulas that are only atomically valid (in particular the double negation law). 
Our Gentzen-style formalism appears optimal in several aspects (Remarks \ref{rem:rex} and \ref{rem:inqb}). 
Further discussion is postponed to Conclusions (Section \ref{sec:conclusions}).
\fi

%\tlnt{A general comment of Referee 3: It would be helpful to be even more explicit about the specific role of schematic
%validity and the Casari formula in the paper}

\section{Preliminaries}
\label{sec:prelim}
Our language $\lang$ assumes a countably infinite set $\mathsf{Var}$ of variables and $\mathsf{Pred}$  of predicate symbols (with fixed arities), respectively, and logical connectives $\bot$, $\land$, $\to$, $\forall$ as well as {\em inquisitive disjunction} $\inqd$ and the {\em inquisitive existential quantifier} $\inqe$: 
%A {\em term} $t$ is a variable or a constant symbol.  
\[
\varphi ::= P(x_{1},\ldots,x_{m}) \,|\,  \bot \,|\, \varphi \to \varphi \,|\, \varphi \land \varphi \,|\, \varphi \inqd \varphi \,|\, \Any{x}\varphi \,|\, \Exi{x}\varphi,
\]
where the arity of $P$ is $m$ (when the arity of $P$ is $0$, $P$ is regarded as a {\em propositional variable}). Throughout the paper, $P(\overline{x})$ denotes $P(x_{1},\ldots,x_{m})$, where $m$ is the arity of $P$.  We define $\neg \varphi$ := $\varphi \to \bot$ and $?\varphi$ := $\varphi \inqd \neg \varphi$. 
By $\varphi[z/x]$, we mean the result of capture-avoiding substitution of free occurrences of $x$ by $z$ in $\varphi$.  
\if0
A formula $\alpha$ not containing any occurrences of $\inqd$ or $\inqe$ is called a {\em classical formula}:
\[
\alpha ::= P(x_{1},\ldots,x_{m}) \,|\,  \bot \,|\, \alpha \to \alpha \,|\, \alpha \land \alpha \,|\, \Any{x}\alpha.
\]
\fi
Note that in order to simplify presentation, we assume that our $\lang$ contains neither identity nor function symbols. An interested reader might compare it with a Rocq formalization by Max Ole Elliger \cite{OleMSc,OleForm}, which allows rigid terms.

%\tlnt{I'd also avoid constants and the whole discussion of rigidity... OTOH, would the presence of constants make a completeness proof simpler?}

%To save space, we keep our presentation of semantics concise and intuitive. We highlight aspects that might seem unfamiliar for the audience familiar with the usual Kriple semantics of superintuitionistic predicate logics. For a detailed presentation, we recommend in particular Ciardelli and Grilletti (see also \dots) \tlnt{refs}.  

A {\em relational information model} is a structure $\mM$ = $(W,D,\mI)$ where $W$ is a non-empty set of possible worlds, $D$ is a non-empty set of individuals and $\mI: \mathsf{Pred} \times W \to D^{< \omega}$ is a function sending each $n$-arity predicate symbol $P$ to an element $\mI(P,w)$ of $\mathcal{P}(D^{n})$. 
For a relational information model $\mM$ = $(W,D,\mI)$, we say that $s \subseteq W$ is a \emph{state}. Let $\mM$ = $(W,D,\mI)$  be a relational information model, $g:\mathsf{Var} \to D$ an assignment and $s \subseteq W$ a state.
We define the notion of support $M, s \models_{g} A$ as follows:
\[
\begin{array}{rcl}
     \mM, s \Vdash_{g} P(x_{1},\ldots, x_{m}) &\mathrm{iff}& (g(x_{1}),\ldots,g(x_{m})) \in \mI(P,w) \text{ for all $w \in s$}; \\
     \mM,s \Vdash_{g} \bot &\mathrm{iff}& s =\varnothing; \\ 
     \mM,s \Vdash_{g} \varphi \land \psi &\mathrm{iff}& \mM,s \Vdash_{g} \varphi \text{ and  } \mM,s \Vdash_{g} \psi; \\   \mM,s \Vdash_{g} \varphi \to\psi &\mathrm{iff}& \text{for all $t \subseteq s$: } \mM,t \Vdash_{g} \varphi \text{ implies } \mM,t \Vdash_{g} \psi; \\
     \mM,s \Vdash_{g} \varphi \inqd \psi &\mathrm{iff}& \mM,s \Vdash_{g} \varphi \text{ or } \mM,s \Vdash_{g} \psi; \\
     \mM,s \Vdash_{g} \Any{x}\varphi &\mathrm{iff}& \mM,s \Vdash_{g[x \mapsto d]} \varphi \text{ for all $d \in D$; } \\
     \mM,s \Vdash_{g} \Exi{x}\varphi &\mathrm{iff}& \mM,s \Vdash_{g[x \mapsto d]} \varphi \text{ for some $d \in D$. } \\
\end{array}
\]
where $g[x \mapsto d]$ is the same mapping as $g$ except that it sends $x$ to $d$. 
%When the underlying model $M$ is clear from the context, we simply write $s \Vdash_{g} \varphi$ instead of $\mM, s \Vdash_{g} \varphi$. 	
%Following the more general conventions introduced above, 
We say that $\varphi$ is \emph{valid} in a relational information model $\mM$ = $(W,D,\mI)$ if 
$\mM, s \Vdash_{g} \varphi$ for every assignment $g: \mathsf{Var} \to D$ and every state $s \subseteq W$. 

%It is easy to obtain the following persistency. 

\begin{proposition}[Persistency]
\label{prop:persistency}
If $\mM,s \Vdash_{g} \varphi$ and $s \supseteq t$ then $\mM,t \Vdash_{g} \varphi$. 
\end{proposition}

\begin{definition}
Define $\inqbq$ to be the set of all the valid formulas in any relational information model $M$. 
We also write $\inqbqn$ for the collection of formulas forced when cardinality of $W$ is restricted to at most $n$, and $\inqbqo = {\bigcap}_{n \in \omega}\inqbqn$. 
\end{definition}

We note that $\neg\neg P(\overline{x}) \to P(\overline{x}) \in \inqbq$ for every $P \in \lna{Pred}$. Clearly, $\inqbq_{1}$ is just the classical logic and for all $n \in \omega$, we have $\inqbq \subseteq \inqbqo \subseteq \inqbqn$.

\section{Schematic Validity and The Casari Axiom} \label{sec:schema} \label{sec:casari}
%xThis section investigates the gap between \emph{schematic validity} of ordinary superintuitionistic logics and theoremhood of (bounded) inquisitive logic.
 The clause for implication and for $\inqd$ given above are those of Kripke models for intuitionistic predicate logic, with the accessibility ordering being the reverse inclusion on states.  %To begin with, note that 
   Obviously, neither $\inqbq$, nor $\inqbqo$, nor $\inqbqn$ for any fixed $n$ is a superintutionistic predicate logic in the sense %required by standard references, which entails closure 
 of being closed under substitutions of formulas for predicate symbols. In order to formalize such substitutions, Ono \cite{OnoH:1973} employs conventions of Church \cite[Ch. III]{Church1958}, %[\S\ III.35, pp. 192--193]
  whereas Gabbay, Shehtman and Skvortsov \cite[\S\ 2.2--2.5]{GabSkvShe} follow Bourbaki in using the notion of a \emph{scheme} (which appears to resemble \emph{locally nameless} representation in mechanical reasoning, cf. \cite{Chargueraud12}). The treatment given by Kleene \cite[\S\ VII.34, pp. 155--162]{klee:intr52} appears compact and accessible. % is particularly compact and readable.
\if0
 but the key idea is clear enough: any Kripke model  forcing $\phi$ % and $\phi$ contains some occurrences of $P \in\lna{Pred}$, 
 also has to force any formula arising from $\varphi$ by properly replacing $P(\overline{x})$ with an arbitrary $\psi(\overline{x})$ satisfying natural (if technically involved) sanity conditions, as one can freely change the interpretation of atoms everywhere in the poset, as long as persistence is respected. 
%This is because we look at all possible elements of $\IntModels{\mN,D}$: if the substituted formula is refuted at $n$ by some $\mI$, we can \emph{freely} define another $\mI' \in \IntModels{\mN,D}$ in which the denotation of $P$ at $n$ is the same as the denotation of $\psi(x_1,\dots,x_m)$ in $\mI$. 
\fi
\begin{definition} \label{def:subst}
Assume  $\bar{a}_1$, $\dots$, $\bar{a}_k$ are $n$-ary lists of variables, $P \in \mathsf{Pred}$ is $n$-ary and $P(\bar{a}_1), \dots, P(\bar{a}_k)$ are all occurrences of $P$ in $\phi$. %(just individual variables in our setting). 
Assume furthermore that $\psi$ is a formula with \emph{exactly} $n$ free variables. Moreover, for each  $i \leq k$,  create an \emph{$i$-clean instance} of $\psi$ by renaming bound variables to ensure no variable in $\bar{a}_i$ becomes captured in $\psi$ after substitution.  Then $(\phi)_{P}^{\psi}$ denotes the outcome of syntactic replacement of each $P(\bar{a}_i)$ with the corresponding $i$-clean instance of $\psi(\bar{a}_i)$.
%Let $(\cdot)_{P}^{\psi}$ be a capture-avoiding substitution that replaces all occurrences of a fixed predicate $P$ uniformly with a fixed formula $\psi$ by adjusting the variables $\overline{x}$ in the arity of $P$ for $\psi$ $($if necessary we rename the bound variables in $\psi$$)$, where we assume the number of free variables in $\psi$ is more than or equal to the arity of $P$. 
\end{definition}

\begin{remark}
%Of course, as there are only finitely many instances in each formula and we have an unlimited supply of variables, we could insist on ensuring. 
A variant of this definition more faithful to that of Kleene  would allow for \emph{at least} $n$ free variables in the formula being substituted instead of \emph{exactly} $n$. In such a situation, we would need to insist on additional clean-up to ensure that these additional variables do not get captured in $\phi$ after substituting corresponding occurrences of $P$. The definition would also get require a minor reformulation if our syntax allowed other individual terms than variables.
\end{remark}

%For a labelled formula $X: \varphi$, we define $(X: \varphi)_{P}^{\psi}$ as $X: (\varphi)_{P}^{\psi}$.
%When $\Gamma$ = $\setof{X_{1}:\varphi_{1},\ldots,X_{n}:\varphi_{n}}$, we define $\Gamma_{P}^{\psi}$ $:=$ $\setof{ (X_{1}:\varphi_{1})_{P}^{\psi},\ldots,(X_{n}:\varphi_{n})_{P}^{\psi}}$. 
%\end{definition}

\begin{remark} \label{rem:sokripke}
Any intuitionistic Kripke structure  forcing some given $\phi$ % and $\phi$ contains some occurrences of $P \in\lna{Pred}$, 
 also has to force $(\phi)_{P}^{\psi}$ for every suitable choice of $P$ and $\psi$, %formula arising from $\varphi$ by properly replacing $P(\overline{x})$ with an arbitrary $\psi(\overline{x})$ satisfying natural (if technically involved) sanity conditions, 
 as one can freely change the interpretation of atomic relations at any node, as long as persistence is respected. However, in a relational information model, the interpretation of atoms  at non-singleton states is induced by their singleton substates.  
Thus, e.g., the atomic validity of double negation elimination does not extend to all its non-atomic instances. \iffull For readers who are not familiar with the relationship between inquisitive logic and intuitionistic logic, we provide a detailed discussion in Appendix \ref{sec:cd}.\else\fi
\end{remark}
%is valid for atoms, it is not valid for arbitrary formulas 
%involving $\inqd$ or $\inqe$. 
%Nevertheless, there are some non-trivial axiom schemes valid in  $\inqbq$, $\inqbqo$ and/or $\inqbqn$. 
Table \ref{table:scheme} present some important examples of schemes to be investigated. 

\newcommand{\negsp}{\vspace{-1pt}}

\begin{table}%[htbp]
\negsp
\caption{Formula Schemes}
\label{table:scheme}
\begin{center}
\begin{tabular}{|c|c|}
\hline
Name & Scheme \\
\hline
\lna{Kuroda} & $\Any{x}\neg\neg\phi \to \neg\neg\Any{x}\phi$ \\
\hline
\lna{CD} & $\Any{x}(\phi(x)\inqd\psi) \to (\Any{x}\phi(x))\inqd\psi)$ where  $x$ does not appear in $\psi$ \\
\hline
{\lna{CasariAtomic}} & $(\Any{x}((P(x) \to \Any{x}P(x)) \to \Any{x}P(x))) \to \Any{x}P(x)$ \\
\hline
{\lna{CasariDNAtomic}} & $(\Any{x}((\neg \neg P(x) \to \Any{x}\neg \neg P(x)) \to \Any{x}\neg \neg P(x))) \to \Any{x}\neg \neg P(x)$ \\
\hline
{\lna{CasariScheme}} & $(\Any{x}((\phi(x) \to \Any{x}\phi(x)) \to \Any{x}\phi(x))) \to \Any{x}\phi(x)$ \\
\hline
\end{tabular}
\end{center}
\negsp
\end{table}

\begin{lemma}[Ciardelli \& Grilletti]
Both $\lna{Kuroda}$ and $\lna{CD}$ are in $\inqbq$. 
\end{lemma} 

\noindent
In this section, we focus on $\lna{CasariAtomic}$ and its schematic version \lna{CasariScheme} from Table \ref{table:scheme}. The following is not hard to establish, using either a chosen Gentzen-style or Hilbert-style system for intuitionistic predicate logic, or one's preferred proof assistant, or favourite semantics such as Kripke frames. \iffull\, A natural deduction derivation is provided in Appendix \ref{sec:casat}.\else\fi.
%\tlnt{What do we do with these Appendices? Shall we vaguely refer to some``full version which will be made available online''?}

\begin{theorem} \label{th:casat}  \label{cor:casat}
Casari scheme instantiated with $($doubly$)$ negated atoms is a theorem of intuitionistic logic: $\lna{CasariDNAtomic} \in \lna{IQC}$. Consequently, $\lna{CasariAtomic} \in \inqbq$.
\end{theorem}
\if0
\begin{proof}
The formula (just like its single-negation variant) in fact lies within the scope of several partial decision procedures for \lna{IQC} such as the \texttt{firstorder} tactic in Coq. 
\qed\end{proof}

\ksnote{I have read the reviewers' comment again, maybe I think it is better to put our comments on coq in the conclusion part. Otherwise we need to explain what is \texttt{firstorder} here. I temporally drop the description on coq. }

\begin{corollary}
$\lna{CasariAtomic} \in \inqbq$.
\end{corollary}

\ksnote{Maybe we can combine Theorem 1 and Corollary 1?}
\fi 

%\ksnote{2nd reivewer commented on a possible relationship of $\lna{CasariAtomic}$ to Loeb axiom, can you add it, Tadeusz?}

\begin{theorem}[Casari \cite{Casari83}, Th. 10(2), p. 294] \label{th:casari}
%$\mN \Vdash 
A Kripke structure for intuitionistic predicate logic forces $\lna{CasariAtomic}$ $($and, consequently, $\lna{CasariScheme}$$)$  %(and consequently $\mN \Vdash \ref{cassc}$) 
whenever it does not contain infinitely ascending chains. 
\end{theorem}

\noindent As mentioned in the introduction, one is reminded of modal propositional schemes sensitive to the presence of infinite chains. %A particularly interesting discussion, 
Esakia \cite{Esakia98:bsl} highlights more such analogies: e.g., between the Kuroda scheme above and the propositional McKinsey scheme.

\begin{corollary}\label{cor:casfinb}
All instances of $\lna{CasariScheme}$  are in $\inqbqo$. 
\end{corollary}

\begin{proof}
For each $n$, the underlying intuitionistic frames of $\inqbqn$ models (with reverse inclusion on states being the partial ordering) do not contain infinitely ascending $\supseteq$-chains, as relational information models induced by finitely many possible worlds contain only finitely many states. 
\qed\end{proof}

\noindent
%For the following semantic characterization of the scheme in question, %may make the above corollary seem quite surprising, highlighting the consequences of restricting the denotations of atoms to those induced by singleton substates. 
%In intuitionistic Kripke structures (as opposed to \emph{models} fixing interpretations of predicate symbols), validity of schemes  is determined by their atomic instances. %However, this is a consequence of excluding all interpretations of atomic predicates not satisfying \ref{regat} at non-singleton states.
%In the inquisitive setting, where denotations of atoms are always closed under double negation, 
%things may get more complex for non-atomic formulas:

\noindent By contrast, in the light of Remark \ref{rem:sokripke},  inquisitively atomic validity and schematic validity may come apart:% And indeed:

%This contrasts with the following result. 

\begin{theorem} \label{th:casfail}
$\lna{CasariScheme}$ with $\phi(x) \equiv \Exi{y} R(x,y)$ fails in \inqbq.  
\end{theorem}

\begin{proof}
    Set $W$ and $D$ to be equal to $\omega$. For a fixed $i \in \omega$, we define $\mI(R,i)$ as follows:
    \begin{align*}
    %\forall m > i. \quad 2m &R^{\mM(i)} 2m \\
    \forall m \in \omega. & \quad ((2m + 1) ,(2m + 1)) \in  \mI(R,i)   \\
    \forall m, j \in \omega. & \quad j \neq i \;\&\; ((j \text{ odd}) \text{ or } (j > 2m)) \; \Longrightarrow (2m,j) \in \mI(R,i)
    %\forall n \geq 1, m \geq i. \quad m &R^{\mM(i)} (m + n) 
    \end{align*}
    %\tlnt{Now just we need one more line}
    %
    %\tlnt{Old attempt:}
    %\begin{align*}
    %%\forall m > i. \quad 2m &R^{\mM(i)} 2m \\
    %\forall m \in \omega. \quad (2  (m + 1)) &R^{\mM(i)} (2  (m + 1))  \\
    %\forall m, i, j \in \omega. \quad 2m < i \& j > i \Rightarrow 2m&R^{\mM(i)} j \\
    %\forall n \geq 1, m \geq i. \quad m &R^{\mM(i)} (m + n) 
    %\end{align*}
    
    \noindent As $D = \omega$, for any $n \in \omega$ and any formula $\psi(x)$ whose sole free variable is $x$ we can write ``$\mM, s \Vdash \psi(n)$'' to denote $\mM, s \Vdash_g \phi(x)$ for some/any valuation s.t. $g(x) = n$.
    
    \begin{claim}
    $\forall s \subseteq \omega, m \in \omega. \quad \mM, s \Vdash \Exi{y} R(2 m + 1,y)$.
    \end{claim}
    This is shown by noting that at every possible state, $y = 2m + 1$ does the job.
    \begin{claim}
    $\forall s \subseteq \omega, m \in \omega. \quad \mM, s \nVdash \Exi{y} R(2 m,y)$ iff $s \supseteq \omega - \{0, 2, \dots, 2m\}$.
    \end{claim}
    
    \noindent {(Proof of \textit{Claim}.)}  To prove the left-to-right implication of this claim, assume that $s \not\supseteq \omega - \{0, 2, \dots, 2m\}$.  This means that there exists $n \not\in s$  which is \begin{center}
    (*) either an odd number or an even number strictly greater than $2m$. 
    \end{center}
    This however means that 
    every $i \in s$ is distinct from $n$ and  hence $(2m,n) \in \mI(R,i)$ by definition, thus $\mM, s \Vdash \Exi{y} R(2 m,y)$.  Conversely, assume $s \supseteq \omega - \{0, 2, \dots, 2m\}$. This means that every $n$ of the form (*) is a member of $s$. Regardless of $i \in \omega$, every $j \in \omega$ such that $(2m,j) \in \mI(R,i)$ must belong to $s$. But whichever $j$ we would pick to witness $\Exi{y} R(2 m,y)$ at $s$, we would get $(2m,j) \notin \mI(R,j)$. \hfill {(QED of \textit{Claim}.)}
    
    These two claims imply that we constructed an ascending-chain-based refutation of the corresponding instance of Casari as needed in Theorem \ref{th:casari}. In more detail, setting $\phi(x)$ to be $\Exi{y} R(x,y)$, by our Claims, entails that the denotation of $\Any{x}\phi(x)$ is the collection 
    $E := \{s \subseteq \omega \mid (\exists n \in \omega.(2n + 1) \not\in s) \vee 
    (\forall m \exists n. 2(m + n) \not\in s)\},$
    i.e., the collection of states that either fail to contain at least one odd number, or whose complement contains infinitely many even numbers. 
    Clearly, $E$ is downward closed, i.e., closed under inclusion. Consequently, we get the following:
    \begin{itemize}
    \item For an odd $m$, the denotation of $\phi(m) \to \Any{x}\phi(x)$ is $E$ and thus the denotation of 
    \iffull\[\else$\fi
    (\phi(m) \to \Any{x}\phi(x)) \to \Any{x}\phi(x)
    \iffull\]\else$\fi
     is the whole space of subsets of $\omega$.
    \item For an even $m$,  the denotation of $\phi(m)$ is the downset of sets not containing at least one element of the form (*), thus evidence for failure of $\phi(m) \to \Any{x}\phi(x)$ is provided by the set
     $$\{s \subseteq \omega \mid (\forall n \in \omega.(2n + 1) \in s) \wedge 
    (\exists m' \forall n. 2(m' + n) \in s) \wedge (\exists n. 2n \not\in s \wedge 2n >m)\},$$ 
    i.e., the set of those states that contain all the odd numbers and cofinitely many even ones (thus refuting $\Any{x}\phi(x)$), but fail to contain at least one even number greater than $2m$ (thus satisfying $\phi(m)$). 
    One easily notes that the superset-closure of this set of states is exactly the complement of $E$, thus the denotation of $\phi(m) \to \Any{x}\phi(x)$ is exactly $E$, and we finish just like in the odd case.
    \end{itemize}
    This shows that the denotation of the antecedent of our instance of $\lna{CasariScheme}$ is the whole state space. It is not the case for the consequent though (which is $E$), from which the theorem follows.
    %To see this, consider all possible scenarios in which $s \not\supseteq \omega - \{0, 2, \dots, 2m\}$.
    %\begin{itemize}
    %\item There might be $n > 2m$ s.t. $n \not\in s$. However, for all $i \neq n$ we have then that 
    %$\mM, i \Vdash 2m R n$. For $i > 2m$, this follows from the second clause of the definition, and for $i \leq 2m$, it follows from the third clause.
    %\item Now assume that $s$ does not contain  some odd $n$ s.t. $n < 2m$. 
    %\end{itemize}
    \qed\end{proof}

%\tlnt{Referee 2: Is it possible that Casari's formula for $\phi(x)$ is valid iff $\phi(x)$ is $n$-coherent for some finite $n$? It is not hard to see that the right-to-left direction holds, but the converse is not obvious to me.}

\begin{remark} \label{rem:cascoh}
A reviewer noted that whenever $\phi(x)$ is \emph{$n$-coherent} \cite{Ciardelli2022a} for some finite $n$, the corresponding Casari instance does not fail and asked if this is in fact a necessary and sufficient condition. The same reviewer also suggested a simplified construction proving Theorem \ref{th:casfail} with
\iffull
\[\mI(R,i) := \{(n, k) \mid i < k \; \text{or}\; (k = n\; \text{and}\; n \neq i)\}.\] 
\else
$\mI(R,i)$ defined as $\{(n, k) \mid i < k \; \text{or}\; (k = n\; \text{and}\; n \neq i)\}$. 
\fi
 Such questions seem promising experiments for a Rocq mechanization developed by Max Ole Elliger \cite{OleMSc,OleForm}. As in its present form it contains a detailed formalization of our proof of Theorem \ref{th:casfail}, we leave experiments with the statement and the proof of the theorem to subsequent versions of the paper or follow-up work. Note that the conjecture regarding $n$-coherence being a \emph{necessary} condition for $\phi(x)$ to yield a non-failing instance appears false, at least  without some minimal additional conditions. %ensuring that $x$ is ``meaningfully used'' and free in $\phi(x)$. 
For example, $\Exi{x}P(x)$ fails to be coherent even in a transfinite sense \cite[Prop 3.3]{Ciardelli2022a}; merely insisting on $x$ being free does not seem a sufficient refinement.  %but this formula would not help us in constructing a suitable counterexample.  %however, we were unable to find a refuting counterexample for the corresponding substitution.
It is generally not obvious whether a refuting instance can be monadic, i.e., involve solely unary predicates, either for the Casari scheme or for other ones that are atomically but not generally valid (we do not have sufficiently many such examples now). It might nevertheless be possible; Rybakov and Shkatov  \cite{RybakovS24} amply illustrate that unary predicates in modal or superintuitionistic predicate logics can often simulate binary ones using the so-called Kripke Trick. %An additional challenge in the inquisitive setting for such constructions is posed by double-negation closure of atomic formulas
\end{remark}

\if0
\begin{question} \label{quest:unary}
Does there exist an instance of Casari involving solely unary predicates that fails in \inqbq? More generally, are there schemes which are atomically valid, but which have refuting unary instances?
\end{question}
\fi

%\ksnote{Shall we incorpolate questions into conclusion?}

%\begin{itemize}
%    \item syntactically. Introduce the notion of schematic core of a logic. Relate it to Medvedev in the propositional case.
%    \item semantically. Note that the logic associated with a class of frames with no restrictions on valuation is schematically closed. 
%\end{itemize}
%
%\tlnt{Alternatively, introduce all this in preliminaries?}

\section{Labelled Sequent Calculus for The Finite-Bounded Case} \label{sec:labseq}

From now on we focus on providing the  proof calculus promised in the introduction. Its schematic aspects are clarified by Theorem \ref{thm:generation_schematic_validities} below. We begin by ``translating''  the satisfaction relation %$\mathfrak{M},s \Vdash_{g} \varphi$ on a relational information model $\mathfrak{M}$ = $(W,D,\mathfrak{J})$ (where $s \subseteq W$ is a state) 
in terms of a labelled formalism, as is common in the literature~\cite{Negri2005,Dyckhoff2011}. 
Since our target is {\em finite}-bounded inquisitive logic $\mathsf{InqBQ}_{n}$ %(the set of all the valid formulas in any relational information model $\mM$ = $(W,D,\mathfrak{J})$ such that $\# W \leqslant n$, 
for fixed but arbitrary $n$, 
we use a non-empty {\em finite} set $X$ of natural numbers as a label, similar to the labelled tableau calculus for propositional dependence logic~\cite{Sano2015a}. We exclude $\varnothing$ from our labels for simplicity. If we impose a constraint $\# X \leqslant n$, our results in the remainder of the paper instantiate to soundness,  completeness and syntactic cut-elimination for $\mathsf{InqBQ}_{n}$.

%\begin{definition}
%Define $\mathsf{InqBQ}$ to be the set of all the valid formulas in any relational information model $\mM$. 
%Define $\mathsf{InqBQ}_{n}$ to be the set of all the valid formulas in any relational information model $\mM$ = $(W,D,I)$ such that $\# W \leqslant n$.  
%\end{definition}

\begin{definition}
A {\em label} $X$ is a non-empty finite subset of $\omega$, i.e., $\varnothing \neq X \subseteq_{\mathrm{fin}} \omega$. A {\em labelled formula} is an expression of the form $X: \varphi$ where $X$ is a label and $\varphi$ is a formula. The {\em length} $|X:\varphi|$ of a labelled formula $X:\varphi$ is the number of all logical connectives in $\varphi$ (including quantifiers), with the cardinality of $X$ being disregarded.  %\tlnt{This needs to be made more precise: quantifiers etc. 
A {\em sequent}, denoted by $\Gamma \Rightarrow \Delta$, is a pair $(\Gamma, \Delta)$ of finite multisets of labelled formulas, where $\Gamma$ is the {\em antecedent} and $\Delta$ is the {\em consequent} of the sequent. 
\end{definition}

\noindent By translating ``$\mathfrak{M}, s \Vdash_{g} \varphi$'' to ``$X: \varphi$'', the satisfaction relation on a relational information model in Section \ref{sec:prelim} naturally provides the axioms and rules in Table \ref{table:seq_calc}. 
Note that the rules $(\Rightarrow \mathtt{at})$ and $(\Rightarrow \to)$ have finitely many premises depending on $\#X$ and the number of all non-empty finite subsets of $X$, respectively.
Since our labels are finite sets of natural numbers, the propositional fragment of our calculus is simpler than the existing labelled G3-style sequent calculi for inquitive propositional logic $\mathsf{InqB}$~\cite{Chen2017,Muller2023} in terms of the number of axioms and rules (see Remark \ref{rem:comparison}). 
This simplicity also enables us to expand the calculus naturally to the first-order level. If we drop the rule $(\Rightarrow \mathtt{at})$, disregard the labels and remove the repetition of the implication formula from the rule $(\to \Rightarrow)$, the resulting system coincides with $\mathbf{G3c}$ \cite[Definition 3.5.1]{TS2000} ~\cite[p.67]{NegriPlato2001}, i.e., G3-style sequent calculus for the first-order classical logic. 

\begin{remark}
\label{rem:comparison}
When the syntax is restricted to the propositional fragment, with the standard assumption that the set of propositional variables (predicate symbols of arity $0$) is countably infinite, we denote by $\mathbf{G}(\mathsf{InqB})$ the propositional fragment of $\mathbf{G}(\mathsf{FBInqBQ})$ where $(\Rightarrow \texttt{at})$ and $(\mathtt{id})$ are of the following form:
\iffull\else\begin{small}\fi
    $$
    \infer[(\Rightarrow \texttt{at})]
    {\Gamma \Rightarrow \Delta, X: P }
    {
    \inset{\Gamma \Rightarrow \Delta, \setof{k}:P}{k \in X}
    }
    \quad
    \infer[(\texttt{id}) \text{ where $X \supseteq Y$}]{X:P, \Gamma \Rightarrow \Delta,Y:P}
    {}.
    $$
\iffull\else\end{small}\fi
%\noindent We should emphasize that our $\mathbf{G}(\mathsf{InqB})$ is simpler labelled sequent calculus for $\mathsf{InqB}$ in the literature (cf.~\cite{Chen2017,Muller2023} for $\mathsf{InqB}$ and \cite{Sano2009a} for  $\mathsf{InqB}_{2}$) in terms of the numbers of axioms and inference rules. 
Since we dispense with atoms such as ``$s \subseteq t$'' 
in our calculus,  we only have 2 initial sequents and 7 logical rules for the propositional fragment\footnote{\label{ft:propcalc} Chen's labelled sequent calculus~\cite{Chen2017} has 3 axioms (initial sequents), 13 logical axioms and 9 rules for relational atoms ``$s \subseteq t$'', whereas that of M\"uller~\cite{Muller2023} has 3 axioms (initial sequents), 10 logical axioms and 11 rules for relational atoms ``$s \subseteq t$''.   The latter reference claims: ``Chen and Ma ~\cite{Chen2017} do not provide a modular construction [...] Our system is also much simpler and allows for a more elegant cut-admissibility proof."~\cite[p.40]{Muller2023}}.
%While we have no rules for relational atoms ``$s \subseteq t$'',
 Note that we have rules with possibly more than or equal to two premises, e.g. $(\Rightarrow \mathtt{at})$, whose number of branches depends on the cardinality of $\# X$ of the label $X$ of the formula $X:P(\overline{x})$ in the conclusion of the rule $(\Rightarrow \mathtt{at})$. 
While this aspect might seem a drawback of the absence of relational atoms, we can conduct the root-first proof search yielding decidability results of our $\mathbf{G}(\mathsf{InqB})$. Note that all of our inference rules are height-preserving invertible and the number of possible  applications of $(\to \Rightarrow)$ is bounded by $2^{\#X} - 1$ where $X:\varphi \to \psi$ is the principal formula of $(\to \Rightarrow)$. By contrast, M\"uller~\cite[pp.65--66]{Muller2023} admits that the algorithm provided in \emph{op.cit.} may generate infinitely many new variables in the course of the proof-search process. 
We come back to $\mathbf{G}(\mathsf{InqB})$ also in Remark \ref{rem:inqb}. 
\end{remark}
%we do not need to have rules for relational atoms corresponding to ``$s \subseteq t$''. See also Remark \ref{rem:inqb}.
%\footnote{For a Hilbert system and a natural deduction system of the propositional inquisitive logic $\mathsf{InqB}$, the reader is referred to~\cite{Ciardelli2010} and~\cite{Ciardelli2022}, respectively. A display calculus for $\mathsf{InqB}$ is also provided in~\cite{Frittella2016}. }

\if0
\begin{remark} \label{rem:oldnd}
For  $\mathsf{InqBQ}_{n}$, Ciardelli and Grilletti~\cite{Ciardelli2022a} provided a sound and strongly complete natural deduction system.
%\footnote{Grilletti~\cite{Grilletti2021} provided a sound and complete natural deduction system for the classical antecedent fragment of first-order inquisitive logic.} 
However, there is little by way of proof-theoretic study of first-order inquisitive logic in terms of sequent calculi. Sano ~\cite{Sano2011c} provides a sound and complete labelled sequent calculus for $\mathsf{InqBQ}_{2}$, %(named $\mathbf{T}\mathsf{InqQL}_{2}$), 
but that calculus is not G3-style and its labels are of a different character than ours. Furthermore, \emph{op.cit.} provided only \emph{semantic} cut elimination. If we impose a constraint $\# X \leqslant 2$, our results below instantiate to soundness,  completeness and syntactic cut-elimination for $\mathsf{InqBQ}_{2}$. 
\end{remark}
\fi

%If we impose a constraint $\# X \leqslant n$, our results below instantiate to soundness,  completeness and syntactic cut-elimination ``for $\mathsf{InqBQ}_{n}$

\begin{table}[htbp]
\hrule
    \centering 
    \iffull\else
    \negsp
    \fi
    \caption{Sequent Calculus $\mathbf{G}(\mathsf{FBInqBQ})$}
    \label{table:seq_calc}
    \iffull\else
    \begin{small}
    \fi
    \[
    \infer[(\mathtt{id}) \text{ where $X \supseteq Y$}]
    {X:P(\overline{x}), \Gamma \Rightarrow \Delta, Y:P(\overline{x})}{}
    \quad
    \infer[(\bot \Rightarrow)]{
    X:\bot, \Gamma \Rightarrow \Delta}{}
    \]
    $$
    \infer[(\Rightarrow \texttt{at})]
    {\Gamma \Rightarrow \Delta, X: P(\overline{x}) }
    {
    \inset{\Gamma \Rightarrow \Delta, \setof{k}:P(\overline{x})}{k \in X}
    }
    $$
\if0
    $$
    \infer[(\texttt{at} \Rightarrow) \text{ where $X \supseteq Y$}]{X:P(x_{1},\dots,x_{n}) , \Gamma \Rightarrow \Delta}
    {Y:P(x_{1},\dots,x_{n}), X:P(x_{1},\dots,x_{n}) , \Gamma \Rightarrow \Delta}
    $$
\fi
    $$
    \infer[(\Rightarrow\land)]{\Gamma \Rightarrow \Delta, X:\varphi \land \psi}{\Gamma \Rightarrow \Delta, X:\varphi & \Gamma \Rightarrow \Delta, X:\psi}
    \quad
    \infer[(\land \Rightarrow)]{X:\varphi \land \psi, \Gamma \Rightarrow \Delta}{X:\varphi, X:\psi, \Gamma \Rightarrow \Delta}
    $$
    $$
    \infer[(\Rightarrow \inqd)]{\Gamma \Rightarrow \Delta, X:\varphi \inqd \psi}{\Gamma \Rightarrow \Delta, X:\varphi, X:\psi}
    \quad
    \infer[(\inqd \Rightarrow)]{X:\varphi \inqd \psi, \Gamma \Rightarrow \Delta}
    {X:\varphi, \Gamma \Rightarrow \Delta & X:\psi, \Gamma \Rightarrow \Delta}
    $$
    $$
    \infer[{{(\Rightarrow \to )}}]
    {\Gamma \Rightarrow \Delta, X:\varphi \to \psi}
    {\inset{Y:\varphi, \Gamma \Rightarrow \Delta, Y:\psi}{X \supseteq Y}}
    $$
    $$
    \infer[{{(\to \Rightarrow)}} \text{ where $X \supseteq Y$}]{X:\varphi \to \psi, \Gamma \Rightarrow \Delta}
    {X:\varphi \to \psi, \Gamma \Rightarrow \Delta, Y:\varphi & Y:\psi, X:\varphi \to \psi, \Gamma \Rightarrow \Delta}
    $$
    $$
    \infer[(\Rightarrow \forall)\dagger]{\Gamma \Rightarrow \Delta, X: \Any{x}\varphi}{\Gamma \Rightarrow \Delta, X: \varphi[z/x]}
    \quad
    \infer[(\forall \Rightarrow)]{X: \Any{x}\varphi, \Gamma \Rightarrow \Delta}{X: \varphi[y/x], X: \Any{x}\varphi, \Gamma \Rightarrow \Delta}
    $$
    $$
    \infer[(\Rightarrow \inqe)]{\Gamma \Rightarrow \Delta, X: \Exi{x}\varphi }{ \Gamma \Rightarrow \Delta, X: \Exi{x}\varphi, X: \varphi[y/x]}
    \quad
    \infer[(\inqe \Rightarrow)\dagger]{X: \Exi{x}\varphi, \Gamma \Rightarrow \Delta}{X: \varphi[z/x], \Gamma \Rightarrow \Delta}
    \qquad
    $$
    where $\dagger$ is the eigenvariable condition: $z$ does not occur in the conclusion. 
    \iffull\else
     \end{small}
     \fi
\hrule
\end{table}

\begin{definition}
A labelled formula not in the context $\Gamma, \Delta$ of the conclusion of a rule is called {\em principal}. 
 A {\em derivation} $\mathcal{D}$ in $\mathbf{G}(\mathsf{FBInqBQ})$ is a finite tree generated from initial sequents $(\mathtt{id})$ and $(\bot \Rightarrow)$ by inference rules of $\mathbf{G}(\mathsf{FBInqBQ})$. The {\em height} of a derivation $\mathcal{D}$ is the length of the longest branch in $\mathcal{D}$. Given a derivation $\mathcal{D}$, $\mathtt{r}(\mathcal{D})$ is the last applied rule $($including the initial sequents$)$ of $\mathcal{D}$. We say that a sequent $\Gamma \Rightarrow \Delta$ is {\em derivable} in $\mathbf{G}(\mathsf{FBInqBQ})$ $($written: $\mathbf{G}(\mathsf{FBInqBQ}) \vdash \Gamma \Rightarrow \Delta$$)$ if there is a derivation $\mathcal{D}$ whose root is $\Gamma \Rightarrow \Delta$. 
    We write $\mathbf{G}(\mathsf{FBInqBQ}) \vdash_{h} \Gamma \Rightarrow \Delta$ to mean that there is a derivation of $\Gamma \Rightarrow \Delta$ whose height is at most $h$. 
\end{definition}

We define the set of subformulas of $\varphi$ as in Troelstra and Schwichtenberg~\cite[Definition 1.1.3]{TS2000} by requiring $\mathtt{Sub}(\Any{x}\psi)$ = $\setof{\Any{x}\psi} \cup \inset{\psi[z/x]}{z \in \mathsf{Var}}$ and similarly for $\mathtt{Sub}(\Exi{x}\psi)$.
Given labelled formulas $X:\varphi$ and $Y:\psi$, we say that {\em $Y:\psi$ is a {\em subexpression} of $X:\varphi$} if $\psi$ is a subformula of $\varphi$ and $Y \subseteq X$. Our calculus is trivially proved to be {\em analytic} in the following sense. 

\begin{proposition}
\label{cor:subfor_prop}
If $\Gamma \Rightarrow \Delta$ is derivable in $\mathbf{G}(\mathsf{FBInqBQ})$ then $\Gamma \Rightarrow \Delta$ has a derivation in $\mathbf{G}(\mathsf{FBInqBQ})$ such that every labelled subformula in the derivation is a subexpression of a labelled formula in $\Gamma, \Delta$.  
\end{proposition}

\begin{example}
Let $X \subseteq \omega$ be finite. To derive $\Rightarrow X: \neg \neg
P(\overline{x}) \to P(\overline{x})$ by the rule $(\Rightarrow \to)$, it suffices to derive a sequent $Y:\neg \neg P(\overline{x}) \Rightarrow Y:P(\overline{x})$ for all $Y$ such that $X \supseteq Y$. 
By the rule $(\Rightarrow \mathtt{at})$, it suffices to show the derivability of 
$Y:\neg \neg P(\overline{x}) \Rightarrow \setof{k}:P(\overline{x})$ for each $k \in Y$. 
Fix any $k \in Y$. Then we proceed as follows:
\iffull\else
\begin{scriptsize}
\fi
\begin{equation*}
\infer[(\to \Rightarrow)]{Y:\neg \neg P(\overline{x}) \Rightarrow \setof{k}:P(\overline{x})}{
\infer[(\Rightarrow \to)]{
Y:\neg \neg P(\overline{x}) \Rightarrow \setof{k}:P(\overline{x}), \setof{k}: \neg P(\overline{x})}{
\infer[(\mathtt{id})]{
\setof{k}:P(\overline{x}), Y:\neg \neg P(\overline{x}) \Rightarrow \setof{k}:P(\overline{x}), \setof{k}:\bot}{}
}
&
\infer[(\bot \Rightarrow)]{\setof{k}:\bot, Y:\neg \neg P(\overline{x}) \Rightarrow \setof{k}:P(\overline{x}) }{}
},
\end{equation*}
\iffull\else
\end{scriptsize}
\fi
where the number of premises of $(\Rightarrow \to)$ above is one. This means that $(\Rightarrow \mathtt{at})$ captures the semantic clause for the atomic formula.  
\end{example}

Recall that the persistency in the semantics on a relational information model is formulated as follows: if $s \supseteq t$ and $\mM,s \Vdash_{g} \varphi$ then $\mM,t\Vdash_{g} \varphi$ 
by Proposition \ref{prop:persistency}. This can be expressed in terms of our sequent calculus:   

\begin{proposition}
\label{prop:persis_seqcalc}
If $X \supseteq Y$, then $X: \varphi, \Gamma \Rightarrow \Delta, Y: \varphi$ is derivable in $\mathbf{G}(\mathsf{FBInqBQ})$ without applying $(\Rightarrow \mathtt{at})$. 
\end{proposition}

\begin{proof}
By induction on $\varphi$, where the rule $(\Rightarrow \mathtt{at})$ is not necessary to conduct an inductive argument.
\qed\end{proof}

%\tlnt{Referee 2: Page 9. In the derivation of CD, some universals are missing on the left of the sequents on the first two lines (these universals don't play a role in the proof, but strictly speaking they have to be there for the rule $\forall\Rightarrow$).}
%
%\tlnt{Referee 3: Presumably one should have to use the fact that $x$ is not free $\psi$ somewhere
%in the derivation; it is not clear where it is used.}

\begin{example} \label{ex:cd}
In order to derive the constant domain axiom by the rule $(\Rightarrow \to)$, it suffices to derive all instances of $Y: \Any{x}(\varphi \inqd \psi) \Rightarrow  Y: (\Any{x}\varphi) \inqd \psi$ (where $x$ does not occur free in $\psi$) for all $Y$ such that $X \supseteq Y$. This is demonstrated as follows.
\iffull\else
\begin{scriptsize}
\fi
\begin{equation*}
\infer[(\Rightarrow \inqd)]{
Y: \Any{x}(\varphi \inqd \psi) \Rightarrow  Y: (\Any{x}\varphi) \inqd \psi
}{
\infer[(\Rightarrow \forall)]{Y: \Any{x}(\varphi \inqd \psi) \Rightarrow  Y: (\Any{x}\varphi), Y:\psi}{
\infer[(\forall \Rightarrow)]{
Y: \Any{x}(\varphi \inqd \psi) \Rightarrow  Y: \varphi[z/x], Y:\psi
}{
\infer[(\inqd \Rightarrow)]{
Y: \varphi[z/x] \inqd \psi, Y: \Any{x}(\varphi \inqd \psi) \Rightarrow  Y: \varphi[z/x], Y:\psi
}{
Y: \varphi[z/x], Y: \Any{x}(\varphi \inqd \psi)  \Rightarrow  Y: \varphi[z/x], Y:\psi
&
Y: \psi, Y: \Any{x}(\varphi \inqd \psi) \Rightarrow  Y: \varphi[z/x], Y:\psi
}
}
}
}
\end{equation*}
\iffull\else
\end{scriptsize}
\fi
\noindent where the left and right topsequents are derivable by Proposition \ref{prop:persis_seqcalc} 
and it is noted for the application of the rule $(\forall \Rightarrow)$ that $Y: (\varphi \inqd \psi) [z/x]$ $\equiv$ $Y: \varphi [z/x] \inqd \psi $ since $x$ does not occur free in $\psi$. 
\if0
% The following is a derivation of the converse direction. 
\begin{small}
\[
\infer[(\Rightarrow \forall)]{
Y: (\Any{x}\varphi) \inqd \psi \Rightarrow Y: \Any{x}(\varphi \inqd \psi) 
}
{
\infer[(\inqd \Rightarrow)]{
 Y: (\Any{x}\varphi) \inqd \psi \Rightarrow Y: \varphi[z/x] \inqd \psi }
{ \infer[(\Rightarrow \inqd)]{Y: (\Any{x}\varphi)  \Rightarrow Y: \varphi[z/x] \inqd \psi}{
\infer[(\forall \Rightarrow)]{Y: (\Any{x}\varphi)  \Rightarrow Y: \varphi[z/x], Y:\psi}{
Y: \varphi[z/x] \Rightarrow Y: \varphi[z/x], Y:\psi
}
} &
\infer[(\Rightarrow \inqd)]{Y: \psi \Rightarrow Y: \varphi[z/x] \inqd \psi }{
Y: \psi \Rightarrow Y: \varphi[z/x], Y:\psi 
}
}
}
\]
\end{small}
where the leftmost and rightmost leaves are derivable by Proposition \ref{prop:persis_seqcalc}.
\fi
\end{example}

Given a relational information model $\mM$ = $(W,D,\mathfrak{J})$, an assignment $g: \mathsf{Var} \to D$ and a mapping $f: \omega\to W$, we define the  satisfaction relation $\mM, f \Vdash_{g} X: \varphi$ for labelled formulas as follows:
\[
\begin{array}{lll}
\mM, f \Vdash_{g} (X: \varphi) &\mathrm{iff}& 
\mM, f[X] \Vdash_{g} \varphi. 
\end{array}
\]
For a sequent $\Gamma \Rightarrow \Delta$, we define $\mM,f \Vdash_{g} \Gamma \Rightarrow \Delta$ as follows: if $\mM,f  \Vdash_{g} (X:\varphi)$ for all $(X:\varphi) \in \Gamma$, then $\mM,f  \Vdash_{g} (Y:\psi)$ for some $(Y:\psi) \in \Delta$. 
        
\begin{proposition}[Soundness]
\label{prop:sound}
If $\mathbf{G}(\mathsf{FBInqBQ}) \vdash \Gamma \Rightarrow \Delta$ then $\mM, f \Vdash_{g} \Gamma \Rightarrow \Delta$ holds for every relational information model $\mM$ = $(W,D,\mathfrak{J})$, every assignment $g: \mathsf{Var} \to D$ and every mapping $f: \omega\to W$.
\end{proposition}
% 
%\noindent       
%\iffull
%An example proof step for the $(\Rightarrow \mathtt{at})$ clause is given in Appendix \ref{sec:proofsound}. 
%\else\fi
%%\tlnt{What do we do with these Appendices? Shall we vaguely refer to some``full version which will be made available online''?}
%\tlnt{I'd argue that this might be a good proof to move back to the paper, there is enough space}
        
\begin{proof}[Sketch]
As an example proof step, we check the validity preservation of $(\Rightarrow \mathtt{at})$. 
Let $\mM$ = $(W,D,\mathfrak{J})$ be a relational information model, $g: \mathsf{Var} \to D$ an assignment and $f: \omega\to W$ be a mapping. Suppose $\mM,f \Vdash_{g} \Gamma \Rightarrow \Delta, \setof{k}:P(\overline{x})$ for all $k \in X$. 
We show that $\mM,f \Vdash_{g} \Gamma \Rightarrow \Delta, X:P(\overline{x})$. 
Assume that $\mM, f \Vdash_{g} (Y:\varphi)$ for all $(Y:\varphi) \in \Gamma$ and 
$\mM, f \not\Vdash_{g} (Z:\psi)$ for all $(Z:\psi) \in \Delta$. 
We show that $\mM, f \Vdash_{g} X:P(\overline{x})$, i.e., 
$\mM, f[X] \Vdash_{g} P(\overline{x})$, which is equivalent to the following: for all $w \in f[X]$, $(g(x_{1}),\ldots, g(x_{n})) \in \mathfrak{J}(P,w)$. So, let us fix $w \in f[X]$. 
We can find a natural number $k \in X$ such that $f(k)$ = $w$. 
By assumption, our supposition $\mM,f \Vdash_{g} \Gamma \Rightarrow \Delta, \setof{k}:P(\overline{x})$ gives us that 
$\mM, f \Vdash_{g} \setof{k}:P(\overline{x})$, which implies our goal.  
\end{proof}       
        
\begin{corollary}
\label{cor:soundness_inqbqn}
If $\mathbf{G}(\mathsf{FBInqBQ}) \vdash \Rightarrow \setof{1,\ldots, n}: \varphi$ then $\varphi \in \mathsf{InqBQ}_{n}$.
\end{corollary}

%\tlnt{Referee 2: Page 10: the notation G(FBInqBQat) occurs, but it does not seem to be defined; I think the authors just mean G(FBInqBQ).}
%\ksnote{Fixed.}

\begin{definition}
Let $(\cdot)_{P}^{\psi}$ be the  substitution operation from Definition \ref{def:subst}. %capture-avoiding substitution that replaces all occurrences of a fixed predicate $P$ uniformly with a fixed formula $\psi$ by adjusting the variables $\overline{x}$ in the arity of $P$ for $\psi$ $($if necessary we rename the bound variables in $\psi$$)$, where we assume the number of free variables in $\psi$ is more than or equal to the arity of $P$. 
For a labelled formula $X: \varphi$, we define $(X: \varphi)_{P}^{\psi}$ as $X: (\varphi)_{P}^{\psi}$.
When $\Gamma$ = $\setof{X_{1}:\varphi_{1},\ldots,X_{n}:\varphi_{n}}$, we define $\Gamma_{P}^{\psi}$ $:=$ $\setof{ (X_{1}:\varphi_{1})_{P}^{\psi},\ldots,(X_{n}:\varphi_{n})_{P}^{\psi}}$. 
\end{definition}

\begin{theorem}
\label{thm:generation_schematic_validities}
If a sequent $\Gamma \Rightarrow \Delta$ is derivable in $\mathbf{G}(\mathsf{FBInqBQ})$ without applying $(\Rightarrow \mathtt{at})$, then 
$\Gamma_{P}^{\psi} \Rightarrow \Delta_{P}^{\psi}$ is also derivable in $\mathbf{G}(\mathsf{FBInqBQ} )$ without applying $(\Rightarrow \mathtt{at})$. 
Therefore, if a sequent $\Rightarrow \setof{1,\ldots, n}: \varphi$ is derivable in $\mathbf{G}(\mathsf{FBInqBQ} )$ without applying $(\Rightarrow  \mathtt{at})$, then $\varphi$ is schematically valid in $\mathsf{InqBQ}_{n}$. 
\end{theorem}

%\ksnote{I just wrote one case for the inductive step, but we have not defined the notion of schematic validity formally so we may just state ``For the inductive step, an argument is straightfoward'', etc. to get more space. }

\begin{proof}
The latter statement follows from the former statement and 
Proposition \ref{prop:sound}. So we focus on proving the former by induction on a derivation with no application of $(\Rightarrow \mathtt{at})$ of $\Gamma \Rightarrow \Delta$. 
For the base case,  it is easy to see that the $(\cdot)_{P}^{\psi}$-translation of an initial sequent $(\bot \Rightarrow)$ is again of the same form $(\bot \Rightarrow)$. 
When $P(\overline{x})$ is not principal of an initial sequent $(\mathtt{id})$, then the $(\cdot)_{P}^{\psi}$-translation of the initial sequent is also of the form $(\mathtt{id})$. 
Otherwise, $X: \psi(\overline{x}), \Gamma_{P}^{\psi} \Rightarrow \Delta_{P}^{\psi}, Y: \psi(\overline{x})$ is derivable without applying $(\Rightarrow \mathtt{at})$ by Proposition \ref{prop:persis_seqcalc}. 
For the inductive step, the argument is straightforward. 
\if0
we only deal with the case where the last applied rule is $(\Rightarrow \forall)$, i.e., from $\Gamma \Rightarrow \Delta, \varphi[z/x]$ we obtain $\Gamma \Rightarrow \Delta, \Any{x}\varphi$ where $z$ is fresh. 
Our goal is to show that $\Gamma_{P}^{\psi} \Rightarrow \Delta_{P}^{\psi}, (\Any{x}\varphi)_{P}^{\psi}$ is derivable without $(\Rightarrow \texttt{at})$. 
Let $(\Any{x} \varphi)_{P}^{\psi} \equiv \Any{x} (\varphi_{P}^{\psi})$ (otherwise we rewrite bound variables to apply a similar argument). Then it suffices to derive 
$\Gamma_{P}^{\psi} \Rightarrow \Delta_{P}^{\psi}, (\varphi_{P}^{\psi}) [u/x]$ without $(\Rightarrow \texttt{at})$ where $u$ is fresh in our goal sequent. 
By height-preserving substitution $[u/z]$, we obtain 
$\Gamma \Rightarrow \Delta, \varphi[u/x]$ and then 
$\Gamma_{P}^{\psi} \Rightarrow \Delta_{P}^{\psi}, (\varphi[u/x])_{P}^{\psi}$ by induction hypothesis.  Since $(\varphi[u/x])_{P}^{\psi} \equiv \varphi_{P}^{\psi} [u/x]$, we obtain the intended derivability of our goal sequent. \fi
\qed
\end{proof}

Whether the converse direction of the second statement of Theorem \ref{thm:generation_schematic_validities} holds is an open question.

%\ksnote{Added one sentence on this.}
%\tlnt{Referee 2:The authors show that is a formula has a proof that does not use the rule $(\Rightarrow \mathtt{at})$, it is schematically valid. Does the converse hold as well? If a formula is schematically valid, is it always possible to prove it without using the rule $(\Rightarrow \mathtt{at})$? If the authors don't have an answer, it might be worth mentioning this as an open problem.}

%%%%%%%%%%%%%%%%%%%%%%%%%%%%%%%%%%%%%%%%%%%%%%%%%%%%%%%%%%%%%%%%%
%%%%%%%%%%%%%%%%%%%%%%%%%%%%%%%%%%%%%%%%%%%%%%%%%%%%%%%%%%%%%%%%
%%%%%%%%%%%%%%%%%%

%\tlnt{How does the proof of soundness and completeness look like?}

\section{Strong Completeness} \label{sec:semcompl}

This section establishes the converse direction of Corollary \ref{cor:soundness_inqbqn} in a stronger form, as the semantic strong completeness.  

\begin{definition}
Let $1 \leqslant n \in \omega$. Given a set $\Psi \cup \setof{\varphi}$ of formulas of $\lang$, we use $\Psi \Vdash_{\mathsf{InqBQ}_{n}} \varphi$ to mean that, for all relational information models $\mM$ = $(W,D,\mI)$ such that $\# W \leqslant n$, all $g:\mathsf{Var} \to D$ and  all $s \subseteq W$, if $\mM, s \Vdash_{g} \psi$ for every $\psi \in \Psi$ then $\mM, s \Vdash_{g} \varphi$.
\end{definition}

\noindent We note that $\varnothing \models_{\mathsf{InqBQ}_{n}} \varphi$ iff $\varphi \in \mathsf{InqBQ}_{n}$. 
In this subsection, $\Gamma$, $\Delta$ are possibly infinite in the expression $\Gamma \Rightarrow \Delta$. 
%In the case where $\Gamma$ and $\Delta$ are all finite, the sequent $\Gamma \Rightarrow \Delta$ said to be {\em finite}. 
A (possibly infinite) sequent $\Gamma \Rightarrow \Delta$
is {\em derivable} in $\mathbf{G}(\mathsf{FBInqBQ})$ 
if $\mathbf{G}(\mathsf{FBInqBQ}) \vdash \Gamma' \Rightarrow \Delta'$ for some finite sequent $\Gamma' \Rightarrow \Delta'$ such that $\Gamma' \subseteq \Gamma$ and $\Delta' \subseteq \Delta$. 

\begin{definition}
\label{dfn:Zsaturation}
Let $Z$ be a label, i.e.,  $\varnothing \neq Z \subseteq_{\mathrm{fin}} \omega$. 
We say that a possibly infinite sequent $\Gamma \Rightarrow \Delta$ is {\em $Z$-saturated} if all the labels 
in $\Gamma$ or $\Delta$ are subsets of $Z$ and the sequent satisfies the following conditions: 
\begin{itemize}[align=left]
    \item[$(\mathtt{unprov})$] $\Gamma \Rightarrow \Delta$ is not derivable in $\mathbf{G}(\mathsf{FBInqBQ})$.
    \item[$(\mathtt{at}R)$] If $X:P(x_{1},\ldots, x_{n}) \in \Delta$ then $\setof{k}:P(x_{1},\ldots, x_{n}) \in \Delta$ for some $k \in X$. 
%    \item[$(\mathtt{at}L)$] If $X:P(x_{1},\ldots, x_{n}) \in \Gamma$ then $Y: P(x_{1},\ldots, x_{n})$ for all $Y \subseteq X$. 
    \item[$(\land R)$] If $X:\varphi_{1} \land \varphi_{2} \in \Delta$ then $X:\varphi_{1} \in \Delta$ or $X:\varphi_{2} \in \Delta$. 
    \item[$(\land L)$] If $X:\varphi_{1} \land \varphi_{2} \in \Gamma$ then $X:\varphi_{1} \in \Gamma$ and $X:\varphi_{2} \in \Gamma$. 
    \item[$(\inqd R)$] If $X:\varphi_{1} \inqd \varphi_{2} \in \Delta$ then $X:\varphi_{1} \in \Delta$ and $X:\varphi_{2} \in \Delta$. 
    \item[$(\inqd L)$] If $X:\varphi_{1} \inqd \varphi_{2} \in \Gamma$ then $X:\varphi_{1} \in \Gamma$ or $X:\varphi_{2} \in \Gamma$.    
    \item[$(\mathop{\to} R)$] If $X:\varphi_{1} \to \varphi_{2} \in \Delta$ then $Y: \varphi_{1} \in \Gamma$ and $Y: \varphi_{1} \in \Delta$ for some $Y \subseteq X$.
    \item[$(\mathop{\to} L)$] If $X:\varphi_{1} \to \varphi_{2} \in \Gamma$ then $Y:\varphi_{1} \in \Delta$ or $Y:\varphi_{2} \in \Gamma$ for all $Y \subseteq X$. 
    \item[$(\forall R)$] If $\Any{x}\varphi \in \Delta$ then $\varphi[z/x] \in \Delta$ for some variable $z$. 
    \item[$(\forall L)$] If $\Any{x}\varphi \in \Gamma$ then $\varphi[y/x] \in \Gamma$ for all variables $y$.
    \item[$(\inqe R)$] If $\Exi{x}\varphi \in \Delta$ then $\varphi[y/x] \in \Delta$ for all variables $y$.
    \item[$(\inqe L)$] If $\Exi{x}\varphi \in \Gamma$ then $\varphi[z/x] \in \Gamma$ for some variable $z$.    
\end{itemize}
\end{definition}

\begin{lemma}
\label{lem:saturation}
Let $\varnothing \neq Z \subseteq_{\mathrm{fin}} \omega$ and $\Gamma \Rightarrow \Delta$ be a possibly infinite sequent. If $\Gamma \Rightarrow \Delta$ is not 
derivable in $\mathbf{G}(\mathsf{FBInqBQ}) $ and all the labels in $\Gamma,\Delta$ are subsets of $Z$, then there exists a $Z$-saturated sequent $\Gamma^{+} \Rightarrow \Delta^{+}$ such that $\Gamma \subseteq \Gamma^{+}$ and $\Delta \subseteq \Delta^{+}$. 
\end{lemma}

\noindent
%Proof proceeds by a Hintikka-style saturation argument. \iffull See Appendix \ref{sec:proofsat}.\else Due to the lack of space, details are postponed to a technical report and a later journal version.\fi
%\tlnt{What do we do with these Appendices? Shall we vaguely refer to some``full version which will be made available online''?}
\begin{proof}
 We proceed by a Hintikka-style saturation argument. Suppose that  $\mathbf{G}(\mathsf{FBInqBQ}) \not\vdash  \Gamma\Rightarrow \Delta$ and that all the labels are subsets of a fixed label $Z$. 
We add a countably infinite fresh set of variables to $\lang$ to define the expanded syntax $\lang^{+}$. 
Let $(X_{i}:\varphi_{i})_{i \in \omega}$ be an enumeration of all labelled formulas such that $\varphi_{i}$ is a formula of $\lang^{+}$ and $X_{i} \subseteq Z$ is a label and each labelled formula $X_{i}:\varphi_{i}$ occurs infinitely many times in the enumeration. Let $(z_{i})_{i \in \omega}$ be an enumeration of all the variables in $\lang^{+}$. In what follows, we inductively define an infinite sequence $(\Gamma_{i} \Rightarrow \Delta_{i})_{i \in \omega}$ of labelled sequents such that each sequent $\Gamma_{i} \Rightarrow \Delta_{i}$ is not derivable in $\mathbf{G}(\mathsf{FBInqBQ})$, $\Gamma_{i} \subseteq \Gamma_{i+1}$, and $\Delta_{i} \subseteq \Delta_{i+1}$ ($i \in \omega$). Then we show that the ``limit'' $\bigcup_{i\in\omega}\Gamma_{i} \Rightarrow \bigcup_{i\in\omega}\Delta_{i}$ of the sequence enjoys all the saturation properties in Definition \ref{dfn:Zsaturation}, $\Gamma  \subseteq \bigcup_{i\in\omega}\Gamma_{i}$ and
$\Delta  \subseteq \bigcup_{i\in\omega}\Delta_{i}$. 
\begin{description}
\item[(Basis)] We put $\Gamma_{0}$ := $\Gamma$ and $\Delta_{0}$ := $\Delta$. It is clear that the sequent $\Gamma_{0} \Rightarrow \Delta_{0}$ is not derivable in $\mathbf{G}(\mathsf{FBInqBQ})$ by assumption. 
\item[(Inductive Step)] Suppose that we have constructed 
$(\Gamma_{i} \Rightarrow \Delta_{i})_{i \leqslant k}$ such that 
each sequent $\Gamma_{i} \Rightarrow \Delta_{i}$ is not derivable in $\mathbf{G}(\mathsf{FBInqBQ})$ ($i \leqslant k$), $\Gamma_{i} \subseteq \Gamma_{i+1}$, and $\Delta_{i} \subseteq \Delta_{i+1}$ ($i < k$).
We define a sequent $\Gamma_{k+1} \Rightarrow \Delta_{k+1}$ in terms of the sequent $\Gamma_{k} \Rightarrow \Delta_{k}$ and the 
$k$ th element $X_{k}: \varphi_{k}$ in the enumeration. The definition is conducted in terms of the following argument by cases. 
\begin{itemize}
\item[$(\mathtt{at}r)$] Let $\varphi_{k}$ be of the form $P(\overline{x})$ and $X_{k}:P(\overline{x}) \in \Delta_{k}$. There is an $m \in X_{k}$ such that
\iffull\[\else$\fi 
\mathbf{G}(\mathsf{FBInqBQ}) \nvdash \Gamma_{k} \Rightarrow \Delta_{k}, \setof{m}: P(\overline{x})
\iffull.\]\else$. \fi 
Suppose otherwise. By the rule $(\Rightarrow \mathtt{at})$, we obtain the derivability of $\Gamma_{k}\Rightarrow \Delta_{k}$, a contradiction. So let us fix such $m \in X_{k}$. Then we define $
\Gamma_{k+1}$ := $\Gamma_{k}$ and  $\Delta_{k+1}$ := $\Delta_{k} \cup \setof{\setof{m}: P(\overline{x})}$.
%
%\item[$(\land r)$] Let $\varphi_{k}$ be of the form $\psi_{1}\land \psi_{2}$ and 
%$X_{k}:\psi_{1}\land \psi_{2} \in \Delta_{k}$. Similarly to $(\inqd l)$.
%
%\item[$(\land l)$] Let $\varphi_{k}$ be of the form $\psi_{1}\land \psi_{2}$ and 
%$X_{k}:\psi_{1}\land \psi_{2} \in \Gamma_{k}$. Similarly to $(\inqd r)$.
%
\item[$(\inqd r)$] Let $\varphi_{k}$ be of the form $\psi_{1}\inqd \psi_{2}$ and 
$X_{k}:\psi_{1}\inqd \psi_{2} \in \Delta_{k}$. 
Then we define \iffull\else\newline\fi $\Gamma_{k+1}$ := $\Gamma_{k} $ and  
\iffull\[\else$\fi
\Delta_{k+1} := \Delta_{k} \cup \setof{X_{k}: \psi_{1}, X_{k}: \psi_{2}}
\iffull.\]\else$.\fi
We can prove the underivability of $\Gamma_{k+1} \Rightarrow \Delta_{k+1}$ by the rule $(\Rightarrow \inqd)$. \newline Clause $(\land l)$ is analogous.
\item[$(\inqd l)$] Let $\varphi_{k}$ be of the form $\psi_{1}\inqd \psi_{2}$ and 
$X_{k}:\psi_{1}\inqd \psi_{2} \in \Gamma_{k}$. 
Either 
\begin{center}
$\Gamma_{k}, X_{k}:\psi_{1} \Rightarrow \Delta_{k}$ \qquad or \qquad 
$\Gamma_{k}, X_{k}:\psi_{2} \Rightarrow \Delta_{k}$ 
\end{center}
is underivable. 
Suppose otherwise. By the rule $(\inqd \Rightarrow)$, we obtain the derivability of 
$\Gamma_{k} \Rightarrow \Delta_{k}$, which is a contradiction. 
Choose one of the underivable sequents as  $\Gamma_{k+1} \Rightarrow \Delta_{k+1}$. \newline Clause $(\land r)$ is analogous. 
\item[$(\mathop{\to} r)$] Let $\varphi_{k}$ be of the form $\psi_{1}\to \psi_{2}$ and $X_{k}:\psi_{1}\to \psi_{2} \in \Delta_{k}$. 
Suppose for contradiction that 
sequents $Y: \psi_{1}, \Gamma_{k}\Rightarrow \Delta_{k}, Y: \psi_{2}$ are derivable for all $Y \subseteq X_{k}$. 
Then the rule $(\Rightarrow \to)$ gives us the derivability of $\Gamma_{k}\Rightarrow \Delta_{k}$, which is a contradiction. Choose one of the underivable sequents as  $\Gamma_{k+1} \Rightarrow \Delta_{k+1}$. 
\item[$(\mathop{\to} l)$] Let $\varphi_{k}$ be of the form $\psi_{1}\to \psi_{2}$ and $X_{k}:\psi_{1}\to \psi_{2} \in \Gamma_{k}$. 
Let $h$ $:=$ $\# (\wp(X_{k}) \setminus \setof{\varnothing})$  and $Z_{1},\ldots,Z_{h} $ be an enumeration of the set. We define an increasing sequence $(\Gamma_{k}^{(j)} \Rightarrow \Delta_{k}^{(j)})_{1 \leqslant j\leqslant h}$ of sequents inductively as follows. 
Let $(\Gamma_{k}^{(1)} \Rightarrow \Delta_{k}^{(1)})$ be \iffull\else\newline\fi $\Gamma_{k} \Rightarrow \Delta_{k}$. 
Suppose that we have defined an increasing sequence \iffull\else\newline\fi $(\Gamma_{k}^{(j)} \Rightarrow \Delta_{k}^{(j)})_{1 \leqslant j\leqslant i}$ of sequents such that each sequent is underivable. We define $\Gamma_{k}^{(i+1)} \Rightarrow \Delta_{k}^{(i+1)}$ as follows. 
Suppose for contradiction that both \iffull\else\newline\fi
$\Gamma_{k}^{(i)} \Rightarrow \Delta_{k}^{(i)}, Z_{j}:\psi_{1}$ and 
$Z_{j}:\psi_{2}, \Gamma_{k}^{(i)} \Rightarrow \Delta_{k}^{(i)}$ are derivable. 
By $X_{k}: \psi_{1} \to \psi_{2} \in \Gamma_{k}$, we obtain the derivability of 
$\Gamma_{k}^{(i)} \Rightarrow \Delta_{k}^{(i)}$ by the rule $(\to \Rightarrow)$. 
Choose one of the underivable sequents as  $\Gamma_{k}^{(i+1)} \Rightarrow \Delta_{k}^{(i+1)}$. 
Define $\Gamma_{k+1} \Rightarrow \Delta_{k+1}$ := 
$\Gamma_{k}^{(h)} \Rightarrow \Delta_{k}^{(h)}$. 
%
%\item[$(\forall r)$] Let $\varphi_{k}$ be of the form $\Any{x}\psi$ and $X_{k}:\Any{x}\psi \in \Delta_{k}$. Similarly to $(\inqe l)$.
%
%\item[$(\forall l)$] Let $\varphi_{k}$ be of the form $\Any{x}\psi$ and $X_{k}:\Any{x}\psi \in \Gamma_{k}$. Similarly to $(\inqe r)$.
%
\item[$(\inqe r)$] Let $\varphi_{k}$ be of the form $\Exi{x}\psi$ and $X_{k}:\Exi{x}\psi \in \Delta_{k}$. We define $\Gamma_{k+1}$ := $\Gamma_{k}$ and 
\iffull\[\else\newline$\fi
\Delta_{k+1} := \Delta_{k} \cup \inset{X_{k}: \psi[z_{i}/x]}{0 \leqslant i \leqslant k}
\iffull.\]\else$. \fi
Suppose for contradiction that \iffull\else\newline\fi $\Gamma_{k+1} \Rightarrow \Delta_{k+1}$ is derivable. By applying $(\Rightarrow \inqe)$ finitely many times, we obtain the derivability of $\Gamma_{k} \Rightarrow \Delta_{k}$, which is a contradiction. 
Therefore, $\Gamma_{k+1} \Rightarrow \Delta_{k+1}$ is not derivable. \newline
Clause $(\forall l)$ is analogous.
\item[$(\inqe l)$] Let $\varphi_{k}$ be of the form $\Exi{x}\psi$ and $X_{k}:\Exi{x}\psi \in \Gamma_{k}$. Since $\Gamma_{k} \Rightarrow \Delta_{k}$ is finite, we can choose the first fresh variable $z$ in the enumeration $(z_{i})_{i \in \omega}$ such that $z$ does not occur free in $\Gamma_{k} \Rightarrow \Delta_{k}$. 
Then we define $\Gamma_{k+1}$ := $\Gamma_{k} \cup \setof{X_{k}:\varphi[z/x]}$ and \iffull\else\newline\fi $\Delta_{k+1}$ := $\Delta_{k}$. We can show that $\Gamma_{k+1} \Rightarrow \Delta_{k+1}$ is not derivable by the rule $(\inqe\Rightarrow)$ since $z$ is a fresh variable not occurring in $\Gamma_{k} \Rightarrow \Delta_{k}$. \newline
Clause $(\forall r)$ is analogous. 
\item[(else)] Otherwise. We put $\Gamma_{k+1}$ := $\Gamma_{k}$ and  $\Delta_{k+1}$ := $\Delta_{k}$. 
\end{itemize}
\end{description}
Finally it is easy to establish that $\bigcup_{i\in\omega}\Gamma_{i} \Rightarrow \bigcup_{i\in\omega}\Delta_{i}$ enjoys all the saturation properties (as well as the condition $(\mathtt{unprov})$) in Definition \ref{dfn:Zsaturation}, $\Gamma  \subseteq \bigcup_{i\in\omega}\Gamma_{i}$ and
$\Delta  \subseteq \bigcup_{i\in\omega}\Delta_{i}$. 
%%\qed\end{proof}

\end{proof}

\begin{definition}
Let  $\varnothing \neq Z \subseteq_{\mathrm{fin}} \omega$ and $\Gamma \Rightarrow \Delta$ be $Z$-saturated. We define the {\em derived relational information model} $\mM_{(\Gamma,\Delta)}$ = $(Z,\mathsf{Var},\mathfrak{J}_{(\Gamma,\Delta)})$ from the sequent as follows:
\begin{itemize}
    \item $(x_{1},\ldots,x_{n}) \in \mathfrak{J}_{(\Gamma,\Delta)}(P,k)$ $\mathrm{iff}$ $\setof{k}: P(x_{1},\ldots,x_{n}) \notin \Delta$. 
\end{itemize}
\end{definition}

\begin{lemma}
\label{lem:truth}
Let  $\varnothing \neq Z \subseteq_{\mathrm{fin}} \omega$ and $\Gamma \Rightarrow \Delta$ be $Z$-saturated. Then the following hold: \begin{center}
$(\mathrm{i})$ if $X: \varphi \in \Gamma$ then $\mM_{(\Gamma,\Delta)}, X \Vdash_{\mathtt{id}} \varphi$;\quad
$(\mathrm{ii})$ if $X: \varphi \in \Delta$ then $\mM_{(\Gamma,\Delta)}, X \not\Vdash_{\mathtt{id}} \varphi$,
\end{center}
where $\mathtt{id}: \mathsf{Var} \to \mathsf{Var}$ is the identity function. 
\end{lemma}

\begin{proof}
We only establish the case where $\varphi \equiv P(\overline{x})$ because the other cases are immediate from the conditions for the $Z$-saturation of $\Gamma \Rightarrow \Delta$. For $(\mathrm{i})$, suppose that $X: P(\overline{x}) \in \Gamma$. 
To show that $\mM_{(\Gamma,\Delta)}, X \models  P(\overline{x})$, let us fix any $k \in X$.  
By definition, it suffices to show that $\setof{k}: P(\overline{x}) \notin \Delta$. 
Suppose for contradiction that $\setof{k}: P(\overline{x}) \in \Delta$. 
It follows from  $X: P(\overline{x}) \in \Gamma$ that  $\Gamma \Rightarrow \Delta$ is an instance of $(\mathtt{id})$, a contradiction with $(\mathtt{unprov})$. 
Next, we move to $(\mathrm{ii})$. We assume that $X: P(\overline{x}) \in \Delta$. 
To show that $\mM_{(\Gamma,\Delta)}, X \not\models  P(\overline{x})$, it suffices to find some $k \in X$ such that $\setof{k}:P(\overline{x}) \in \Delta$. 
This holds by $(\mathtt{at}R)$. \qed\end{proof}

\noindent Write $\setof{1,\ldots,n}:\Psi$ for $\inset{\setof{1,\ldots,n}: \psi}{\psi \in \Psi}$. Now we can state:

\begin{theorem}
\label{thm:semcompl}
Fix any $n \in \omega$ and let $\Psi \cup \setof{\varphi}$ be a possibly infinite set of formulas. If $\Psi \Vdash_{\mathsf{InqBQ}_{n}} \varphi$ then $\mathbf{G}(\mathsf{FBInqBQ})  \vdash \setof{1,\ldots,n}:\Psi \Rightarrow \setof{1,\ldots,n}: \varphi$.  In particular, if $\varphi \in \inqbq_{n}$ then $\Rightarrow \setof{1,\ldots, n}: \varphi$ is derivable in $\mathbf{G}(\mathsf{FBInqBQ})$.
\end{theorem}

%\ksnote{Contrary to the second reviewer's expectations based on an analogy to the Loeb axiom, we can prove strong completeness I think. My method of proving the completeness of $\inqbq_{2}$ is in terms of the canonical model for first-order intuitionistic logic with constant domains. Thus, we can prove the strong completeness for $\inqbq_{2}$. However, if we require the condition $W = D$ for our frame class, then proving strong completeness might be difficult. }

\begin{proof}
Let $n \in \omega$ and $Z$ := $\setof{1,\ldots,n}$. 
Suppose that $Z:\Psi \Rightarrow Z:\varphi$ is not derivable in $\mathbf{G}(\mathsf{FBInqBQ})$. By Lemma \ref{lem:saturation}, we can find a $Z$-saturated sequent $\Gamma \Rightarrow \Delta$ such that $Z:\Psi  \subseteq \Gamma$ and $Z:\varphi \in \Delta$. By Lemma \ref{lem:truth}, we obtain $\mM_{(\Gamma, \Delta)}, Z \Vdash_{\mathtt{id}} \psi$ for all $\psi \in \Psi$ but $\mM_{(\Gamma, \Delta)}, Z \not\Vdash_{\mathtt{id}} \varphi$ hence $\Psi \not\models_{\mathsf{InqBQ}_{n}} \varphi$, as desired.  
\qed\end{proof}

\begin{remark}
\label{rem:strongcomp}
In the natural deduction system for $\mathsf{InqBQ}_{n}$ given by Ciardelli and Grilletti~\cite{Ciardelli2022a}, a key r\^ole is played by the $n$-cardinality formula depending on the signature $\mathsf{Pred}$~\cite[Section 6]{Ciardelli2022a}. They also established the strong completeness of the natural deduction system for $\mathsf{InqBQ}_{n}$ in~\cite[Theorem 8.1]{Ciardelli2022a}. 
In our labelled sequent calculus, the cardinality restriction is simply captured by the signature-independent label $\setof{1,\ldots, n}$ in Corollary \ref{cor:soundness_inqbqn} and Theorem~\ref{thm:semcompl}. 
%Given the analogy between the Casari axiom scheme and the L{\"o}b axiom, it might be surprising that there are sound and {\em strongly complete} proof systems for $\mathsf{InqBQ}_{n}$. This is not the case, because, even if a possibly infinite $\Psi$ is consistent in terms of $\mathsf{InqBQ}_{n}$, it is satisfiable in a relational information model $\mM$ $=$ $(W,D,\mI)$ such that $\#W \leqslant n$ and it is impossible to have infinite ascending chains in $\mM$. 
\end{remark}

%\tlnt{ Referee 2:
%On page 12, the authors mention that G(FBInqBQ) can also provide a proof system for the rex fragment of InqBQ. Since this is, in my view, an important point, I would like the authors to discuss this point a bit more precisely.}

%\tlnt{Referee 2:  in the definition of the rex fragment, inquisitive disjunction is missing}

\begin{remark}
    \label{rem:rex}
    %Ciardelli and Grilletti \cite{Ciardelli2022a} define 
    The {\em restricted existential} (\emph{rex}) fragment $\langrex$ \cite{Ciardelli2022a}  is defined by: 
    \[
    \chi ::= P(\overline{x}) \,|\, \bot \,|\, \chi \land \chi \,|\, \chi \inqd \chi \,|\, \varphi \to \chi | \Any{x} \chi,
    \]
    where $\varphi$ is a formula in $\lang$. %and $\chi$ is referred to as a {\em rex fromula}. 
    Theorem 9.2 of Ciardelli and Grilletti \cite{Ciardelli2022a} provides a complete system for the entailments with rex conclusions in terms of a {\em coherence rule} involving the above-mentioned signature-dependent cardinality formula. However,  characterizing entailments with rex conclusions in $\mathbf{G}(\mathsf{FBInqBQ})$ does not require additional inference rules.  
    First, it is a consequence of Theorem 4.3 in \emph{op.cit.} that entailments with rex conclusions are compact. Second, Propositions 3.5 and 5.3 in \emph{op.cit.} imply that, for every $\chi \in \langrex$, $\chi \in \mathsf{InqBQ}$ iff there exists $n_{\chi} \in \omega$ such that $\chi \in \mathsf{InqBQ}_{n_{\chi}}$, where $n_{\chi}$ is given by Definition 5.2 in \emph{op.cit} and it is recursively computed from $\chi \in \langrex$. 
    Now, we can show that %, our calculus $\mathbf{G}(\mathsf{FBInqBQ})$ can characterize the entailments with rex conclusions: 
    for every $\Sigma \subseteq \lang$ and $\chi \in \langrex$, the following are equivalent: 
    \begin{itemize}
    \item[$(\mathrm{i})$] $\Sigma \models \chi$;  
    \item[$(\mathrm{ii})$] $\mathbf{G}(\mathsf{FBInqBQ}) \vdash \,\Rightarrow \{1,\ldots,n \}: \bigwedge \Sigma' \to \chi$, for some  finite subset $\Sigma' \subseteq \Sigma$ and $n \in \omega$. 
    \end{itemize}
    We can prove the equivalence as follows. 
    By Theorem 4.3 in \emph{op.cit.}, $\Sigma \models \chi$ iff $\models \bigwedge \Sigma' \to \chi$ for some finite subset $\Sigma' \subseteq \Sigma$. 
    Fix such $\Sigma'$. Because 
    $\bigwedge \Sigma' \to \chi$ is still (equivalent to) a rex formula, %by~\cite[Definition 5.2]{Ciardelli2022a}, 
    $\bigwedge \Sigma' \to \chi \in \mathsf{InqBQ}$ iff there exists $n \in \omega$ such that $\bigwedge \Sigma' \to \chi \in \mathsf{InqBQ}_{n}$ by the observation above. By Corollary \ref{cor:soundness_inqbqn} and Theorem~\ref{thm:semcompl}, it is equivalent to $\mathbf{G}(\mathsf{FBInqBQ}) \vdash \,\Rightarrow \{1,\ldots,n \}: \bigwedge \Sigma' \to \chi$, as desired. The converse direction is immediate. 
    %Observe that this characterization of the entailments with rex conclusions does not require extending $\mathbf{G}(\mathsf{FBInqBQ})$, while 
    %Ciardelli and Grilletti~\cite[Section 9]{Ciardelli2022a} the one given in \S\ 9 of \emph{op.cit.} requires a special {\em coherence rule} involving the above-mentioned signature-dependent cardinality formula. 
\end{remark}

%%%%%%%%%%%%%%%%%%%%%%%%%%%%%%%%%%%%%%%%%%%%%%%%%%%%%%%%%%%
%%%%%%%%%%%%%%%%%%%%%%%%%%%%%%%%%%%%%%%%%%%%%%%%%%%%%%%%%%%

\section{Cut Elimination and Derivability of Schematic Validities} \label{sec:cutelim}

%\tlnt{Referee 2: the notion of height-preserving invertibility is not defined.}
%\ksnote{Added.}

First, this section establishes the admissibility of the cut rule on the height-preserving admissibility of weakening rules, inversions of all inference rules of $\mathbf{G}(\mathsf{FBInqBQ})$ and contraction rules. Second, as an application of the admissibility of structural rules and the cut rule, we provide derivations for $\lna{CasariScheme}$ as well as key axioms underlying the natural deduction system for $\mathsf{InqBQ}_{n}$ (see Remark \ref{rem:rex}), except for the signature-specific cardinality formula~\cite[Section 6]{Ciardelli2022a}.

\begin{definition}
We say that an inference rule 
\[
\infer[(\mathtt{r})]{\Gamma_{0} \Rightarrow \Delta_{0}}{\Gamma_{1}\Rightarrow \Delta_{1} &\cdots & \Gamma_{m}\Rightarrow \Delta_{m}}
\]
is {\em height-preserving admissible} in $\mathbf{G}(\mathsf{FBInqBQ})$ if whenever $\mathbf{G}(\mathsf{FBInqBQ}) \vdash_{h} \Gamma_{i} \Rightarrow \Delta_{i}$ for all $1 \leqslant i \leqslant m$ we obtain $\mathbf{G}(\mathsf{FBInqBQ}) \vdash_{h} \Gamma_{0} \Rightarrow \Delta_{0}$. 
When we drop the subscript ``$h$'' from all the occurences of ``$\vdash_{h}$'', we say that $(\mathtt{r})$ is {\em admissible}. 
\end{definition}

\noindent By a standard argument ~\cite[Ch. 3.5]{TS2000}~\cite{NegriPlato2001}, we get the following.

%\tlnote{The reference is incomplete} 
%\ksnote{What information is missing? I usually use "Section" in this case, as far as the authors seem to use \texttt{section} command. }

% definition of substitution
\begin{proposition}
\label{prop:admissible_rules}
\begin{itemize}
    \item[$(\mathrm{i})$] If $\mathbf{G}(\mathsf{FBInqBQ}) \vdash_{h} \Gamma \Rightarrow \Delta$ then $\mathbf{G}(\mathsf{FBInqBQ}) \vdash_{h} \Gamma[z/x] \Rightarrow \Delta[z/x]$, where 
    \iffull\[\else$\fi
    \Theta[z/x] := \inset{X:\varphi[z/x]}{X:\varphi \in \Theta}
     \iffull\]\else.$\fi
    \item[$(\mathrm{ii})$] The weakening rules are height-preserving admissible in $\mathbf{G}(\mathsf{FBInqBQ})$:
    \[
    \infer[(\Rightarrow w)]{\Gamma \Rightarrow \Delta, X:\varphi}{\Gamma \Rightarrow \Delta}
    \quad
    \infer[(w \Rightarrow)]{X:\varphi, \Gamma \Rightarrow \Delta}{\Gamma \Rightarrow \Delta}.
    \]
    \item[$(\mathrm{iii})$] All the rules of $\mathbf{G}(\mathsf{FBInqBQ})$ are height-preserving invertible, i.e., 
    whenever its conclusion of a rule of $\mathbf{G}(\mathsf{FBInqBQ})$ is derivable, the premises of the same rule are also derivable with no greater derivation height.    
    \item[$(\mathrm{vi})$] The contraction rules are height-preserving admissible in $\mathbf{G}(\mathsf{FBInqBQ})$:
    \[
    \infer[(\Rightarrow c)]{\Gamma \Rightarrow \Delta, X:\varphi}{\Gamma \Rightarrow \Delta, X:\varphi, X:\varphi}
    \quad
    \infer[(c \Rightarrow)]{X:\varphi, \Gamma \Rightarrow \Delta}{X:\varphi, X:\varphi, \Gamma \Rightarrow \Delta}.
    \]  
\end{itemize}
\end{proposition}

\begin{theorem}
\label{thm:cut_elim}
The rule of cut is admissible in $\mathbf{G}(\mathsf{FBInqBQ})$:
$$
\infer[(Cut)]{\Gamma, \Pi \Rightarrow \Delta, \Sigma}{\Gamma \Rightarrow \Delta, X:\varphi & X:\varphi, \Pi \Rightarrow \Sigma}.
$$
where $X:\varphi$ is said to be a {\em cut formula}. 
\end{theorem}

\begin{proof}
Let $\mathcal{L}$ of $\Gamma \Rightarrow \Delta, X:\varphi$ and $\mathcal{R}$ of $X: \varphi, \Pi \Rightarrow \Sigma$ be derivations in $\mathbf{G}(\mathsf{FBInqBQ})$. Our proof is by induction on the lexicographic order of the length $\mathtt{c}$ of the cut labelled formula $X:\varphi$ and the cut-height $\mathtt{h}$, i.e., the sum of the height of the derivation $\mathcal{L}$ and the height of the derivation $\mathcal{R}$. 
We follow the standard argument in G3-style sequent calculus to divide our argument into the following three cases.
\begin{enumerate}
    \item One of $\mathtt{r}(\mathcal{L})$ and $\mathtt{r}(\mathcal{R})$ is an initial sequent $(\mathtt{id})$ or $(\bot \Rightarrow)$. 
    \item $\mathtt{r}(\mathcal{L})$ or $\mathtt{r}(\mathcal{R})$ is a rule where the cut formula is not principal in the rule. 
    \item \label{caseiii} $\mathtt{r}(\mathcal{L})$ and $\mathtt{r}(\mathcal{R})$ are rules where the cut formula is principal in both rules. 
\end{enumerate}
We comment on the difference from the standard argument for G3-style sequent calculi ~\cite{NegriPlato2001,Negri2005,Dyckhoff2011}. 
In particular, we focus on the cases related to the initial sequent $(\mathtt{id})$ and the new rule $(\Rightarrow \mathtt{at})$. 
First, we comment on Case 1 for the initial sequent $(\mathtt{id})$ of the form $X:P(\overline{x}), \Gamma \Rightarrow \Delta, Y: P(\overline{x})$ where $X \supseteq Y$, where the principal formulas are $X:P(\overline{x})$ and $Y:P(\overline{x})$. 
Let $\mathtt{r}(\mathcal{L})$ or $\mathtt{r}(\mathcal{R})$ be $(\mathtt{id})$. 
%\tlnt{Referee 3: I am not entirely sure what is going on here: immediately below there appears
%to be a case distinction for subcases of this one, but one of the cases is
%just (r(L),R(R)) = ((id), (id)) again, and the other conflicts with the above:
%(r(L),R(R)) = (($\Rightarrow$ at), (id)). Either there is a typo, or I am not following the structure of this proof.}
%\ksnote{I revised the sentence "Let $\mathtt{r}(\mathcal{L})$ and $\mathtt{r}(\mathcal{R})$ be $(\mathtt{id})$" into 
%"Let $\mathtt{r}(\mathcal{L})$ \textcolor{red}{or} $\mathtt{r}(\mathcal{R})$ be $(\mathtt{id})$." }
When the cut formula is not principal in $(\mathtt{id})$, it is immediate. 
So we assume that the cut formula is principal in $(\mathtt{id})$. 
Then we need to consider the following two cases: 
(i)  $(\mathtt{r}(\mathcal{L}), \mathtt{r}(\mathcal{R}))$ = $((\mathtt{id}),(\mathtt{id}))$, 
(ii) $(\mathtt{r}(\mathcal{L}), \mathtt{r}(\mathcal{R}))$ = $((\Rightarrow \mathtt{at}),(\mathtt{id}))$. 
For the case (i), it suffices to consider the following transformation: 
\begin{gather*}
    \infer[(Cut)]{X:P(\overline{x}), \Gamma,\Pi \Rightarrow \Delta, \Sigma, Z:P(\overline{x})}{
    \infer[(\mathtt{id})]{
    X:P(\overline{x}), \Gamma \Rightarrow \Delta, Y:P(\overline{x})
    }{}
    &
    \infer[(\mathtt{id})]{
    Y:P(\overline{x}), \Pi \Rightarrow \Sigma, Z:P(\overline{x})
    }{}
    }
    \\
    \leadsto 
    \infer[(\mathtt{id})]{
    X:P(\overline{x}), \Gamma,\Pi \Rightarrow \Delta, \Sigma, Z:P(\overline{x})
    }{}
\end{gather*}    
where $X \supseteq Y \supseteq Z$. 
For case (ii), we have the following derivation: 
\[
\infer[(Cut)]{
\Gamma,\Pi \Rightarrow \Delta, \Sigma, Y:P(\overline{x})
}{
\infer[(\Rightarrow\mathtt{at})]{\Gamma \Rightarrow \Delta, X:P(\overline{x})}
{
\inset{ \Gamma \Rightarrow \Delta, \setof{k}:P(\overline{x}) }{k \in X}
}
&
\infer[(\mathtt{id})]{
X:P(\overline{x}), \Pi \Rightarrow \Sigma, Y:P(\overline{x})}
{}
}
\]
where $X \supseteq Y$. By $X \supseteq Y$, we can obtain the following cut-free derivation:
\[
\infer=[(\Rightarrow w), (w \Rightarrow)]{
\Gamma,\Pi \Rightarrow \Delta, \Sigma, Y:P(\overline{x})
}{
\infer[(\Rightarrow \mathtt{at})]{
\Gamma \Rightarrow \Delta, Y:P(\overline{x})}
{
\inset{\Gamma \Rightarrow \Delta, \setof{k}:P(\overline{x})}{ k \in Y}
}
}.
\]
Second, we consider Case 2, in particular when $\mathtt{r}(\mathcal{L})$ or $\mathtt{r}(\mathcal{R})$ is $(\Rightarrow \mathtt{at})$. Here, the cut formula is not principal in $(\Rightarrow \mathtt{at})$. Then we can lift the application of $(Cut)$ for the conclusion of $(\Rightarrow \mathtt{at})$ 
to an application of $(Cut)$ for the premise. %of $(\Rightarrow \mathtt{at})$.
 The lifted application of $(Cut)$ can be eliminated since the complexity of the cut formula is the same but the cut height of the rewritten derivation becomes smaller than the original derivation. 
\qed\end{proof}

\begin{corollary}
\label{cor:mon_seqcalc}
The following rule is admissible in $\mathbf{G}(\mathsf{FBInqBQ})$: let $X \supseteq Y$. 
\[
\infer[(\mathtt{Mon})]{X:\varphi, \Gamma \Rightarrow \Delta}{Y:\varphi, \Gamma \Rightarrow \Delta}.
\]
\end{corollary}

\begin{proof}
Let $X \supseteq Y$. Assume that 
$Y:\varphi, \Gamma \Rightarrow \Delta$ is 
in $\mathbf{G}(\mathsf{FBInqBQ})$. 
By Proposition \ref{prop:persis_seqcalc}, $X: \varphi \Rightarrow Y: \varphi$ is derivable in $\mathbf{G}(\mathsf{FBInqBQ})$. 
By the admissibility of the cut rule (Theorem \ref{thm:cut_elim}), $\mathbf{G}(\mathsf{FBInqBQ}) \vdash X:\varphi, \Gamma \Rightarrow \Delta$ holds. 
\qed\end{proof}

\begin{proposition} \label{prop:dercasari}
For every finite subset $X \subseteq \omega$, the following sequent 
\[
X: \Any{x}((\phi(x) \to \Any{x}\phi(x)) \to \Any{x}\phi(x)) \Rightarrow  X:  \Any{x}\phi(x)
\]
is derivable in $\mathbf{G}(\mathsf{FBInqBQ})$. Therefore, a sequent $\Rightarrow X: \lna{CasariScheme}$ is derivable in $\mathbf{G}(\mathsf{FBInqBQ})$, 
for every finite subset $X \subseteq \omega$. 
\end{proposition}

\begin{proof}
The latter is obtained from the former by applying the rule $(\Rightarrow \to)$. 
Thus we show the former statement by induction on $\# X$. 
For the base step, we assume $X$ = $\setof{n}$. Then we can regard that all the inference rules with the label $\setof{n}$
are the same as those for the classical first-order logic, and so  the sequent in question
%$$\setof{n}: \Any{x}((\phi(x) \to \Any{x}\phi(x)) \to \Any{x}\phi(x)) \Rightarrow  \setof{n}: \Any{x}\phi(x)$$ 
is easily derivable in $\mathbf{G}(\mathsf{FBInqBQ})$. 
For the inductive step, let $\# X > 1$. 
%\tlnt{Referee 3: It seems to me that the right topsequent should also contain $X ? \forall x.\phi(x)$ in
%the antecedent (in order for the application of the rule $\forall\rightarrow$ to be a correct one).}
%\ksnote{Fixed.}
\iffull\else
\begin{scriptsize}
\fi
\[
\infer[(\Rightarrow \forall)]{X: \Any{x}((\phi(x) \to \Any{x}\phi(x)) \to \Any{x}\phi(x)) \Rightarrow  X:  \Any{x}\phi(x)}{
\infer[(\forall \Rightarrow)]{X: \Any{x}((\phi(x) \to \Any{x}\phi(x)) \to \Any{x}\phi(x)) \Rightarrow  X:  \phi(z_{1})}{
\infer[(\to \Rightarrow )]{
X: \Any{x}((\phi(x) \to \Any{x}\phi(x)) \to \Any{x}\phi(x)), X: (\phi(z_{1}) \to \Any{x}\phi(x)) \to \Any{x}\phi(x) \Rightarrow  X:  \phi(z_{1})
}{
\infer[(\Rightarrow \to)]{\Theta \Rightarrow X: \phi(z_{1}) \to \Any{x}\phi(x), X:  \phi(z_{1})}{
\inset{\Theta, Y: \phi(z_{1}) \Rightarrow Y: \Any{x}\phi(x), X:  \phi(z_{1}) }{X \supseteq Y}
}
&
\infer[(\forall \Rightarrow)]{\Theta, X: \Any{x}\phi(x)\Rightarrow X:  \phi(z_{1})}{
\Theta, X: \phi(z_{1}), X: \Any{x}\phi(x)\Rightarrow X:  \phi(z_{1})
}
}
}
}
\]
\iffull\else
\end{scriptsize}
\fi
where \begin{footnotesize}$\Theta := X: \Any{x}((\phi(x) \to \Any{x}\phi(x)) \to \Any{x}\phi(x)), X: (\phi(z_{1}) \to \Any{x}\phi(x)) \to \Any{x}\phi(x).$\end{footnotesize}
Since the right topsequent is derivable by Proposition \ref{prop:persis_seqcalc}, we focus on providing a derivation for a sequent 
\iffull\[\else$\fi
\Theta, Y: \phi(z_{1}) \Rightarrow Y: \Any{x}\phi(x), X:  \phi(z_{1})
\iffull\]\else$\fi
for each $Y$ such that $X \supseteq Y$. Fix any $Y$ such that $X \supseteq Y$. Suppose $Y = X$, then the sequent is derivable by Proposition \ref{prop:persis_seqcalc}. So, let us assume that $Y \neq X$, which implies $\#Y < \#X$. 
By induction hypothesis, the sequent $Y: \Any{x}((\phi(x) \to \Any{x}\phi(x)) \to \Any{x}\phi(x)) \Rightarrow  Y:  \Any{x}\phi(x)$ is derivable. By Corollary \ref{cor:mon_seqcalc}, $X: \Any{x}((\phi(x) \to \Any{x}\phi(x)) \to \Any{x}\phi(x)) \Rightarrow  Y:  \Any{x}\phi(x)$ is derivable. By the admissibility of weakening rules (by Proposition~\ref{prop:admissible_rules} (ii)), we obtain our desired derivability of  $\Theta, Y: \phi(z_{1}) \Rightarrow Y: \Any{x}\phi(x), X:  \phi(z_{1})$. 
\qed\end{proof}

\noindent
\iffull
The reader can find more examples of derivations %formulas and schemes %key axioms from~\cite{Ciardelli2022a} also 
in Appendix \ref{sec:otherkeyaxioms}. 
\else\fi %More examples of derivations will be provided in the full version.\fi

\tlnt{Move one-two more derivations to the body of the paper?}
%\tlnt{What do we do with these Appendices? Shall we vaguely refer to some``full version which will be made available online''?}

\begin{remark}
\label{rem:inqb}
This is a continuation of Remark \ref{rem:comparison} on the propositional fragment $\mathbf{G}(\mathsf{InqB})$. 
Our syntactic argument for Theorem \ref{thm:cut_elim} implies the admissibility of cut rule in $\mathbf{G}(\mathsf{InqB})$.  %Note again that in the propositional case, $\mathsf{InqB}$  coincides with $\mathsf{InqL}$ \cite{Ciardelli09}. 
We can prove the following soundness and completeness for $\mathsf{InqB}$, where we regard $\mathsf{InqB}$ as the set of valid (propositional) formulas for the propositional reduction of our semantics (we do not need to consider the domain $D$). The following two clauses are equivalent: (i) $\mathbf{G}(\mathsf{InqB}) \vdash \,\Rightarrow \setof{1,\ldots,2^{n}}:\varphi$; 
(ii) $\varphi \in \mathsf{InqB}$ (note again that propositionally $\mathsf{InqB}$  coincides with $\mathsf{InqL}$ \cite{Ciardelli09}), where $n$ is the number of propositional variables in $\varphi$. 
The direction from (i) to (ii) is already established by  Proposition \ref{prop:sound}. 
For the direction from (ii) to (i), we proceed as follows. 
Suppose the negation of (i), i.e., $\mathbf{G}(\mathsf{InqB}) \nvdash \,\Rightarrow \setof{1,\ldots,2^{n}}:\varphi$. We prove $\varphi \notin \mathsf{InqB}$. Put $Z$ = $\setof{1,\ldots,2^{n}}$. By Lemma \ref{lem:saturation}, we can find a (propositionally) $Z$-saturated (finite) labelled sequent $\Gamma \Rightarrow \Delta$ such that $\setof{1,\ldots,2^{n}}:\varphi \in \Delta$. We define our derived model $\mathfrak{M}_{(\Gamma,\Delta)}$ = $(Z,v)$ as: $k \in v(P)$ iff $\setof{k}:P \notin \Delta$ (where $P$ occurs in $\Gamma$ or $\Delta$, otherwise $v(P)$ := $\varnothing$, we use ``$v$'' instead of $\mathfrak{J}$ here). 
By Lemma \ref{lem:truth}, we obtain $\mathfrak{M}_{(\Gamma,\Delta)}, \setof{1,\ldots,2^{n}} \nVdash \varphi$ hence $\varphi \notin \mathsf{InqB}$. 
\end{remark}

\if0
\begin{question} \label{quest:craig}
Do $\inqbq$, $\inqbqo$ and $\inqbq_{n}$ for a fixed $n$ enjoy the Craig interpolation property?
\end{question}
\fi

%\section{Open Questions}

\section{Conclusions} \label{sec:conclusions}

Our results reveal complex interplay between correspondence theory and proof theory of inquisitive logic. %and highlight the intricacies of relational information countermodels.
 We have seen, for example, that the status of the Casari axiom in the base inquisitive logic resembles that of the double negation law: it is valid atomically %(Theorem \ref{cor:casat}), 
 but not schematically. %(Theorem \ref{th:casfail}). 
 Its schematic validity is characteristic of $\inqbqo$. Remark \ref{rem:cascoh} discusses further challenges and open questions related to these results. %(Corollary \ref{cor:casfinb}). 

Further intriguing questions regarding notions closely related to boundedness such as \emph{coherence} %(not discussed here due to the lack of space) 
are raised by Ciardelli and Grilletti \cite[\S\ 10]{Ciardelli2022a}. Several cardinality-related questions posed therein appear related to papers concerning model theory and correspondence theory of extensions of \lna{CD} \cite{MinariEA90,OnoH:1973}. It is not known, for example, if \inqbq\ is complete with respect to relational information models with at most \emph{countable} collection $W$ of possible worlds. It would be also of interest to see if the notion of schematic validity and our labelled sequent apparatus can be useful in resolving Ciardelli and Grilletti's challenge of algorithmically identifying formulas coherent for a fixed cardinality. Remark \ref{rem:rex} hints at the relevance of our sequent setup for their question whether the coherence rule is indispensable for entailments with rex conclusions.

Given our proof-theoretic results, it appears natural to consider the Craig interpolation property for logics considered herein. %(Question \ref{quest:craig} concluding Section \ref{sec:cutelim}),
 It is also natural to mechanize our calculus in a proof assistant; some spadework in the Rocq proof assistant covering partially our results (but without, e.g., syntactic cut elimination) has been done by Max Ole Elliger \cite{OleMSc,OleForm}. %\footnote{An example use of this formalization would be to verify the correctness of a simplified construction proving Theorem \ref{th:casfail} suggested by one of the referees with $\mI(R,i)$ defined as $\{(n, k) \mid i < k \; \text{or}\; (k = n\; \text{and}\; n \neq i)\}$.} 
 Furthermore, it could be of interest to extend existing computational interpretations of sequent calculi \cite[Ch. 7]{SorensenU06:book} \cite{CurienH00} to our setting. Continuing the topic of CS applications, the discussion of Armstrong relations  by  Abramsky and V{\"{a}}{\"{a}}n{\"{a}}nen \cite{AbramskyV09} or Ciardelli's perspective on \emph{mention-some} questions  \cite[p. 348 and Ch. 4.7]{Ciardelli2016} suggest relationships with database theory that appear rather under-explored by both communities.

\paragraph{Acknowledgements} We thank the reviewers of several versions of this paper for their constructive comments. %We would also like to acknowledge members of Informatik 8 at FAU Erlangen-Nuremberg and ASTREA group at the University of Naples Federico II for support and discussions at various stages of write-up. 
The work of the first named author has been partially funded by the PNRR MUR project FAIR (No. PE0000013-FAIR).  
The work of the second named author was partially supported by JSPS KAKENHI Grant-in-Aid for Scientific Research (B) Grant Number JP 22H00597 and (C) Grant Number JP 25K03537. 
The work of both authors was partially supported by JSPS KAKENHI Grant-in-Aid for Scientific Research  (C) Grant Number JP 19K12113.
We would also like to acknowledge the feedback of Max Ole Elliger at  FAU Erlangen-Nuremberg; his work on formalization was particularly helpful in the context of Section \ref{sec:casari}.

\bibliographystyle{aiml}
\bibliography{vsano20200911,tadeuszinq,negmod}

%%%%%%%%%%%%%%%%%%%%%%%%%%%%%%%%%%%%%%%%%%%%%%%%%%%%%%%%%%%%
%%%%%%%%%%%%%%%%%%%%%%%%%%%%%%%%%%%%%%%%%%%%%%%%%%%%%%%%%%%%
%%%%%%%%%%%%%%%%%%%%%%%%%%%%%%%%%%%%%%%%%%%%%%%%%%%%%%%%%%%%

\iffull

\appendix

\section{Constant Domain Sematics and Inquisitive Logic} \label{sec:cd}
We recall the broader definition of Kripke semantics for superintuitionistic predicate logic of constant domains, of which the inquisitive semantics is a special case. This might be useful for some readers not familiar with the background of Section \ref{sec:schema}. Fix a nonempty domain $D$ of individual elements and a partially ordered collection $\mN := \tuple{N,\leq}$ of \emph{nodes}. To simplify presentation, we assume the existence of the greatest (inconsistent) node $max^\mN$, and at least one node distinct from $max^\mN$.

\begin{itemize}
\item Define $\mathsf{MOD}_{D}$ to be the class of all classical FO-structures on $D$ for \lna{Pred}. 
%\item Let $W$ be the class of all \emph{non-empty} subsets of $\mathsf{MOD}_{D}$. 
%\item For any $s,s' \in W_D$, $s \alert{\leq} s'$ $\iff$ $s \supseteq s'$. 
\item Let $\mathfrak{J}$ be a function from $N$ to  $\mathsf{MOD}_{D}$ s.t for every $P \in \lna{Pred}$ and for every $n,o \in N$ s.t. $n \leq o$, $P^{\mathfrak{J}(n)} \subseteq P^{\mathfrak{J}(o)}$. The collection of all such functions is denoted as \IntModels{\mN,D}.
 Relatively to a given valuation of individual variables $g : \mathsf{Var} \to D$, the notion of forcing is defined in the standard intuitionistic way: %$\nng \varphi$ is defined in a standard way, with $\inqd$ being the intuitionistic disjunction and $\inqe$ being the intuitionistic existential quantifier. Formally:

\begin{tabular}{L{4.25cm}@{\;if\quad}L{7.7cm}} 
$\nng \bot$ & $n = max_\mathfrak{N}$,  \\[1mm]
$\nng P(x_1,\dots,x_m)$ & $\tuple{g(x_1),\dots,g(x_m)} \in P^{\mathfrak{J}(n)}$,  \\[1mm]
$\nng \phi \wedge \psi$  & $\nng \phi$ and $\nng \psi$, \\[1mm]
$\nng \phi \to \psi$  & $\forall o \geq n, \ndm, o \Vdash_g \phi  \text{ implies } \ndm, o \Vdash_g \psi$,  \\[1mm]
$\nng \phi \inqd \psi$  & $\nng \phi$ or $\nng \psi$, \\[1mm]
$\nng \Any{x}\phi$ & for all $d \in D$, $\ndm, n \Vdash_{g[x \mapsto d]} \phi$, \\[1mm]
$\nng \Exi{x}\phi$ & exists $d \in D$ s.t. $\ndm, n \Vdash_{g[x \mapsto d]} \phi$.
\end{tabular}
%\item Then, $\tuple{W_{D},\supseteq,D}$ is our intended \emph{constant-domain} Kripke model. We define `$s, g \Vdash A$' as usual: 
%\[
%\begin{array}{lll}
%s, g \Vdash P(x_1,\ldots,x_n) &\iff& \tuple{g(x_1),\ldots,g(x_n)} \in P^V_s \\
%s, g \Vdash \Any{x} A &\iff& \text{for all $d \in D$: } s, g(x|d) \Vdash A.\\
%\end{array}
%\]
\end{itemize}

\noindent Here, $g[x \mapsto d]$ is the same mapping as $g$ except that it sends $x$ to $d$. The apparent ``classicality'' of the clause for the universal quantifier is obviously related to the constant domain assumption, %; one can easily show that $\Vdash_g$ is monotone wrt $\leq$. 
which yields the following:

\begin{proposition}[CD Persistency]
\label{prop:persgen}
If $o \geq n$ and $\nng \varphi$ then $\nog \varphi$. 
\end{proposition}

\noindent
Furthermore, $\ndm \Vdash_g \varphi$ if $\ndm, n \Vdash_g \varphi$ for all $n \in N$, moreover $\ndm \Vdash \varphi$ if $\ndm \Vdash_g \varphi$ for all $g$, also $\mathfrak{N}, D \Vdash \varphi$ if $\ndm \Vdash \varphi$ for all $\mathfrak{J} \in \IntModels{N,D}$ and finally $\mathfrak{N} \Vdash \varphi$ if $\mN, D \Vdash \varphi$ for all $D$. We specialize this constant domain semantics to inquisitive semantics as follows.

\begin{itemize}
\item Fix a non-empty collection  $W$ of \emph{possible worlds}. We call any subset $s \subseteq W$ a \emph{state} (relative to $W$). %The collection of all $W$-states is denoted as $\states{W}$.
%\item Define $\mathsf{MOD}_{D}$ to be the class of all classical FO-structures on $D$ for \lna{Pred}. 
%\item Let $W$ be the class of all \emph{non-empty} subsets of $\mathsf{MOD}_{D}$. 
For any $s,s' \in \powset(W)$, $s \leq s'$ is defined as $s \supseteq s'$. In other words, $\states{W} := \tuple{\powset(W),\supseteq}$ is our ordered collection of nodes.
\item Let $\mM$ be a function from $W$ to  $\mathsf{MOD}_{D}$. %That is, we assume that possible worlds correspond to models over $D$. 
The class of all such functions is denoted as \InjModels{W,D}.\footnote{Ciardelli's original MSc Thesis \cite{Ciardelli09} required injectivity of such functions. In newer references, this requirement is dropped. It sometimes leads to complications such as considering \emph{essential cardinality} instead of ordinary cardinality \cite[\S 6]{Ciardelli2022a}, but these considerations do not arise in the present paper. Instead, the injectivity restriction could slightly complicate notationally our completeness proof in Section \ref{sec:semcompl}.}
 Any $\mM \in \InjModels{W,D}$ can be equivalently seen as a \emph{relational information model} $(W,D,\mI)$, where $\mI: \mathsf{Pred} \times W \to D^{< \omega}$ is a function sending each $n$-arity predicate symbol $P$ to $\mI(P,w) \in\mathcal{P}(D^{n})$. 
\item Any $\mM \in \InjModels{W,D}$ is lifted to $\lft{\mM} \in \IntModels{\states{W},D}$ by
%\item For any such state including the maximal state $W$, the family of its substates (non-empty subsets) with reverse inclusion is obviously a poset and $\mathfrak{M}$ can be lifted to yield a model of intuitionistic predicate logic of \emph{constant domains} \tlnt{refs} with the interpretation of any $P \in \mathsf{Pred}$ given by %, given a valuation $g$ of individual variables, is
\begin{equation} \label{regat}
\tag{\lna{RegAt}} P^{\lft{\mM}(s)} := \bigcap \inset{P^{\mathfrak{M}(w)}}{w \in s}. 
\end{equation}
%Define
%$\fwdm := \tuple{\states{W},\supseteq,D,\mathfrak{M}}$ to be the induced constant domain model. Relatively to a given valuation of individual variables $g : \mathsf{Var} \to D$, the notion of forcing $\fwdm, g \Vdash \varphi$ is defined in a standard way, with $\inqd$ being the intuitionistic disjunction and $\inqe$ being the intuitionistic existential quantifier. Formally:
%
%\begin{tabular}{L{5cm}@{\quad if\quad}L{7.2cm}} 
%$\fwdm, g, s \Vdash P(\overline{x})$ & $g(x_1),\dots,g(x_n) \in P^s$  \\[1mm]
%$\fwdm, g, s \Vdash \phi \to \psi$  & $\forall s' \subseteq s, \fwdm,s' \Vdash \phi$ implies \newline \quad\qquad $\fwdm,s' \Vdash \psi$  \\[1mm]
%$\wdm, g, s \Vdash \phi \inqd \psi$  & $\wdm,s \Vdash \phi$ or $\wdm,s \Vdash \psi$  \\[1mm]
%\end{tabular}
\end{itemize}
%
%Furthermore, $\fwdm \Vdash \varphi$ if $\fwdm,g \Vdash \varphi$ for all $\phi$, and $\mathfrak{F}(W,D) \Vdash \varphi$ if $\fwdm \Vdash \varphi$ for all $\mathfrak{M} \in \InjModels{W,D}$. 
When the underlying model $\mM$ is clear from the context, we simply write $s \Vdash_{g} \varphi$ instead of $\lft{\mM}, s \Vdash_{g} \varphi$. We can also omit the lifting notation and write $\mM, s \Vdash_{g} \varphi$ instead, as this does not cause any confusion, thus arriving at the definition of support given in Section \ref{sec:prelim} as reformulation.  %A {\em relational information model} is a structure $M$ = $(W,D,I)$ where $W$ is a non-empty set of possible worlds, $D$ is a non-empty set of individuals and $I: \mathsf{Pred} \times W \to D^{< \omega}$ is a function sending each $n$-arity prediate symbol $P$ to an element $I(P,w)$ of $\mathcal{P}(D^{n})$. 
%For a relational information model $M$ = $(W,D,I)$, we say that $s \subseteq W$ is a state. Let $M$ = $(W,D,I)$  be a relational information model, $g:\mathsf{Var} \to D$ an assignment and $s \subseteq W$ a state.
%We define the notion of {\em support} $M, s \Vdash_{g} A$ as follows:

\section{A ND Derivation Supporting Theorem \ref{th:casat}}
\label{sec:casat}

%A natural deduction derivation is provided as follows. 
%\begin{footnotesize}
\[
\infer[(\to I)_{3}]{
(\Any{x}((\neg\neg P(x) \to \Any{x}\neg\neg P(x)) \to \Any{x}\neg\neg P(x))) \to \Any{x}\neg\neg P(x) 
}{
\infer[(\forall I)]{\Any{x}\neg \neg P(x)}{
\infer[(\neg I)_{2}]{\neg \neg P(z)}{
\infer[(\neg E)]{\bot}{
[\neg P(z)]_{2}
&
\infer[(\forall E)]{\neg\neg P(z)}{
\infer[(\to E)]{
\Any{x}\neg\neg P(x)
}{
\infer[(\to I)_{1}]{
\neg \neg P(z) \to \Any{x} \neg \neg P(x)
}{
\infer[(\bot)]{\Any{x} \neg \neg P(x)}{
\infer[(\neg E)]{\bot}{
[\neg P(z)]_{2}& [\neg \neg P(z)]_{1}
}
}
}
&
\infer[(\forall E)]{(\neg\neg P(z) \to \Any{x}\neg\neg P(x)) \to \Any{x}\neg\neg P(x) }{[\Any{x}((\neg\neg P(x) \to \Any{x}\neg\neg P(x)) \to \Any{x}\neg\neg P(x))]_{3} }
}
}
}
}
}
}
\]

\section{Examples of Derivations}
\label{sec:otherkeyaxioms}

\begin{corollary}
\label{cor:neg_seqcalc}
The following two rules are admissible in $\mathbf{G}(\mathsf{FBInqBQ})$: 
\[
\infer[(\Rightarrow \neg)]{\Gamma \Rightarrow \Delta, X:\neg\varphi}{\inset{\setof{k}:\varphi, \Gamma \Rightarrow \Delta}{k \in X}},
\quad
\infer[(\neg \Rightarrow) \text{ where $X \supseteq Y$,}]{X:\neg\varphi, \Gamma \Rightarrow \Delta}{\Gamma \Rightarrow \Delta, Y:\varphi}
\]
\end{corollary}

\noindent  The proof involves admissibility of weakening rules and Corollary \ref{cor:mon_seqcalc}. 
%Details can be found in Appendix \ref{sec:negseqcalc}.

\begin{proof}
Because the admissibility of the rule $(\neg \Rightarrow)$ is easy, we focus on the admissibility of the rule $(\Rightarrow \neg)$. 
Suppose that $\setof{k}:\varphi, \Gamma \Rightarrow \Delta$ is derivable for each $k \in X$. 
To show the derivability of 
$\Gamma \Rightarrow \Delta, X:\neg\varphi$, it suffices to establish the derivability of 
$Y:\varphi, \Gamma \Rightarrow \Delta, Y:\bot$ for all $X \supseteq Y$. 
Fix any $Y$ such that $X \supseteq Y$. 
Since $Y$ is nonempty, fix some $k \in Y$. 
By the initial supposition, 
$\setof{k}:\varphi, \Gamma \Rightarrow \Delta$ is derivable. 
Since $\setof{k} \subseteq Y$, Corollary \ref{cor:mon_seqcalc} implies that $Y:\varphi, \Gamma \Rightarrow \Delta$ is derivable. 
By the admissibility of weakening rules (by Proposition~\ref{prop:admissible_rules} (ii)), 
we obtain our desired derivability of $Y:\varphi, \Gamma \Rightarrow \Delta, Y:\bot$. \qed
\end{proof}

\begin{proposition} \label{prop:kurocalc}
For every finite subset $X \subseteq \omega$, a sequent $X:\Any{x}\neg\neg\varphi \Rightarrow X:\neg \neg \Any{x}\varphi$ is derivable in $\mathbf{G}(\mathsf{FBInqBQ})$. Therefore, a sequent $\Rightarrow X: \lna{Kuroda}$ is derivable in $\mathbf{G}(\mathsf{FBInqBQ})$, 
for every finite subset $X \subseteq \omega$. 
\end{proposition}

%\noindent Proof involves Corollary \ref{cor:neg_seqcalc}. 
%Details can be found in Appendix \ref{sec:kurocalc}.

\begin{proof}
It suffices to establish the first claim. 
In what follows, we use the admissibility of $(\Rightarrow \neg)$ and $(\neg \Rightarrow)$ from Corollary \ref{cor:neg_seqcalc}.
%\begin{small}
\[
\infer[(\Rightarrow \neg)]{X:\Any{x}\neg\neg\varphi \Rightarrow X:\neg \neg \Any{x}\varphi}{
\infer[(\neg \Rightarrow)]{
\inset{ \setof{k}:\neg \Any{x}\varphi, X:\Any{x}\neg\neg\varphi \Rightarrow }{k \in X}
}{
\infer[(\Rightarrow \forall)]{\inset{X:\Any{x}\neg\neg\varphi \Rightarrow  \setof{k}:\Any{x}\varphi}{k\in X}}{
\infer[(\forall \Rightarrow)]{ \inset{X:\Any{x}\neg\neg\varphi \Rightarrow  \setof{k}:\varphi[z/x]}{k\in X} }{
\infer[(\neg\Rightarrow)]{\inset{X: \neg\neg \varphi [z/x] \Rightarrow  \setof{k}:\varphi[z/x]}{k\in X}}
{
\infer[(\Rightarrow \neg)]{\inset{\Rightarrow  \setof{k}:\varphi[z/x], \setof{k}: \neg \varphi [z/x]}{k \in X} }
{
\inset{\setof{k}: \varphi [z/x] \Rightarrow  \setof{k}:\varphi[z/x]}{k\in X}
}
}
}
}
}
}
\]
%\end{small}
where the topmost sequents are derivable by Proposition \ref{prop:persis_seqcalc}. 
\qed\end{proof}

\begin{proposition}
\label{prop:KP}
A sequent $X: \neg \theta \to (\varphi \inqd \psi) \Rightarrow X: (\neg \theta \to \varphi) \inqd (\neg \theta \to \psi)$ is derivable in $\mathbf{G}(\mathsf{FBInqBQ})$ for every finite subset $X \subseteq \omega$.
\end{proposition}

\noindent Proof involves the rule $(\to \Rightarrow)$, Proposition \ref{prop:persis_seqcalc}, the admissibility of weakening rules, and Corollary \ref{cor:mon_seqcalc}. %Details can be found in Appendix \ref{sec:propKP}.

\begin{proof}
From the bottom, we may proceed as follows. 
%\begin{small}
\[
\infer[(\Rightarrow \inqd)]{X: \neg \theta \to (\varphi \inqd \psi) \Rightarrow X: (\neg \theta \to \varphi) \inqd (\neg \theta \to \psi)}{
\infer[(\Rightarrow \to)]{X: \neg \theta \to (\varphi \inqd \psi) \Rightarrow X: (\neg \theta \to \varphi),  X: (\neg \theta \to \psi)}{
\infer[(\Rightarrow \to)]{\inset{X: \neg \theta \to (\varphi \inqd \psi), Y: \neg \theta \Rightarrow Y: \varphi,  X: (\neg \theta \to \psi)}{X \supseteq Y}}{
\inset{X: \neg \theta \to (\varphi \inqd \psi), Y: \neg \theta, Z: \neg \theta, \Rightarrow Y: \varphi,  Z: \psi}{X \supseteq Y, Z}
}
}
}
\]
%\end{small}
To provide a derivation of each of the topmost sequents, fix any $Y, Z \subseteq \omega$ such that $Y, Z \subseteq X$. Then we can provide our desired derivation as follows. 
By the rule $(\to \Rightarrow)$ and the admissibility of weakening rules, it suffices to derive the following two sequents. 
%\begin{multicols}{2}
\begin{itemize}
    \item $Y: \neg \theta, Z: \neg \theta \Rightarrow Y \cup Z: \neg \theta$, %\columnbreak
    \item $Y \cup Z:\varphi \inqd \psi \Rightarrow Y: \varphi,  Z: \psi$. 
\end{itemize}
%\end{multicols}
\noindent The second sequent is derivable by the rule $(\inqd \Rightarrow)$ and Corollary \ref{cor:mon_seqcalc}. 
The first sequent is derivable as follows:
%\begin{small}
\[
\infer[(\Rightarrow \neg)]{Y: \neg \theta, Z: \neg \theta \Rightarrow Y \cup Z: \neg \theta}{
\infer[(\neg \Rightarrow)]{
\inset{ \setof{k}: \theta, Y: \neg \theta, Z: \neg \theta \Rightarrow }{k \in Y}}{
\inset{ \setof{k}: \theta, Z: \neg \theta \Rightarrow \setof{k}: \theta}{k \in Y}
}
&
\infer[(\neg \Rightarrow)]{
\inset{ \setof{k}: \theta, Y: \neg \theta, Z: \neg \theta \Rightarrow }{k \in Z}}{
\inset{ \setof{k}: \theta, Y: \neg \theta \Rightarrow \setof{k}: \theta}{k \in Z}
}
}
\]
%\end{small}
where all the leaves are derivable by Proposition \ref{prop:persis_seqcalc}.
\qed\end{proof}

\begin{example}
\label{ex:EK}
$X: \Exi{x}\varphi \Rightarrow X: \Exi{x}(\neg \theta \to \varphi)$ is derivable in $\mathbf{G}(\mathsf{FBInqBQ})$ as:
%\begin{small}
\[
\infer[(\inqe\Rightarrow)]{
X: \Exi{x}\varphi \Rightarrow X: \Exi{x}(\neg \theta \to \varphi)
}{
\infer[(\Rightarrow \inqe)]{
X:\varphi[z/x] \Rightarrow X: \Exi{x}(\neg \theta \to \varphi)
}{
\infer[(\Rightarrow \to)]{
X:\varphi[z/x] \Rightarrow X: \Exi{x}(\neg \theta \to \varphi),  X: \neg \theta[z/x] \to \varphi[z/x]
}{
\inset{X:\varphi[z/x], Y: \neg \theta[z/x] \Rightarrow X: \Exi{x}(\neg \theta \to \varphi),   Y: \varphi[z/x]}{X \supseteq Y}
}
}
}
\]
%\end{small}
where the topmost sequents are derivable by Proposition \ref{prop:persis_seqcalc}.
\end{example}

\begin{proposition}
\label{prop:EKP}
A sequent $X: \neg \theta \to \Exi{x} \varphi(x)  \Rightarrow X: \Exi{x} (\neg \theta \to \varphi(x))$ is derivable in $\mathbf{G}(\mathsf{FBInqBQ})$ for every finite subset $X \subseteq \omega$.
\end{proposition}

\noindent
A rather complex proof involves Proposition \ref{prop:persis_seqcalc}, Example \ref{ex:EK}, the admissibility of weakening rules, and the rule $(\Rightarrow \neg )$ of Corollary \ref{cor:neg_seqcalc}. 
%See Appendix \ref{sec:proofEKP}.

\begin{proof}
When $\# X$ = $1$, we can regard all the inference rules with singleton labels are the same as those for the classical first-order logic, the sequent is easily derivable in in $\mathbf{G}(\mathsf{FBInqBQ})$. 
In what follows, we assume $\# X > 1$. For simplicity, we establish the derivability of the sequent when $\# X$ = 2. Put $X$ = $\setof{1,2}$. To derive 
$$\setof{1,2}: \neg \theta \to \Exi{x} \varphi(x)  \Rightarrow \setof{1,2}: \Exi{x} (\neg \theta \to \varphi(x))$$ by the rules
$(\Rightarrow \inqe)$ and 
$(\Rightarrow \to)$, it suffices to derive the following sequents:  
%\begin{small}
\begin{itemize}
    \item[(i)] $\setof{1,2}: \neg \theta \to \Exi{x} \varphi(x), \setof{1,2}: \neg \theta \Rightarrow \setof{1,2}:\varphi(z_{0}), \setof{1,2}:\Exi{x} (\neg \theta \to \varphi(x))$. 
    \item[(ii)] $\setof{1,2}: \neg \theta \to \Exi{x} \varphi(x),  \setof{k}: \neg \theta \Rightarrow  \setof{k}: \varphi(z_{0}), \setof{1,2}:\Exi{x} (\neg \theta \to \varphi(x))$ where $k \in \setof{1,2}$. 
\end{itemize}
%\end{small}
where $z_{0}$ is a fresh variable. 
First, we establish the former, i.e., clause (i).
By $(\to \Rightarrow)$ and admissibility of weakening rules, it suffices to derive both of the following:
%\begin{small}
\begin{itemize}
    \item[(i-a)] $\setof{1,2}: \neg \theta \Rightarrow \setof{1,2}: \neg \theta$, 
    \item[(i-b)] $\setof{1,2}: \Exi{x} \varphi(x) \Rightarrow \setof{1,2}: \Exi{x} (\neg \theta \to \varphi(x))$.
\end{itemize}
%\end{small}
(i-a) is clearly derivable by Proposition \ref{prop:persis_seqcalc}. 
For (i-b), it is derivable by Example \ref{ex:EK}. This finishes establishing (i). 

Now we move to establish the derivability of (ii).
We assume without loss of generality that $k$ = $1$ and then the desired sequent is:
\[
\setof{1,2}: \neg \theta \to \Exi{x} \varphi(x),  \setof{1}: \neg \theta \Rightarrow  \setof{1}: \varphi(z_{0}),\setof{1,2}: \Exi{x} (\neg \theta \to \varphi(x))
\]
By $(\to \Rightarrow)$ and admissibility of weakening rules, 
it suffices to prove the following two sequents:
%\begin{small}
\begin{itemize}
    \item[(ii-a)] $\setof{1,2}: \neg \theta \to \Exi{x} \varphi(x),  \setof{1}: \neg \theta \Rightarrow  \setof{1,2}: \Exi{x} (\neg \theta \to \varphi(x)), \setof{1}: \neg \theta$, 
    \item[(ii-b)] $\setof{1}: \Exi{x} \varphi(x), 
\setof{1,2}: \neg \theta \to \Exi{x} \varphi(x),  \setof{1}: \neg \theta \Rightarrow \setof{1,2}: \Exi{x} (\neg \theta \to \varphi(x))$,
\end{itemize}
%\end{small}
where (ii-a) is clearly derivable by Proposition \ref{prop:persis_seqcalc}. Thus we establish (ii-b) below. 
By the rules
$(\inqe \Rightarrow)$, $( \Rightarrow \inqe)$  and $(\Rightarrow \to)$, it suffices to derive the following sequents:
%\begin{small}
\[
\setof{1}: \varphi(z_{2}), 
\setof{1,2}: \neg \theta \to \Exi{x} \varphi(x),  \setof{1}: \neg \theta, Z: \neg \theta \Rightarrow \setof{1,2}: \Exi{x} (\neg \theta \to \varphi(x)), Z: \varphi(z_{2})
\]
%\end{small}
for all $Z \subseteq \setof{1,2}$, where $z_{2}$ is a fresh variable. 
When $Z$ = $\setof{1}$, it is clearly derivable by Proposition \ref{prop:persis_seqcalc}. 
When $Z$ = $\setof{1,2}$, the sequent is derivable similarly to (i). 
So we consider the case of $Z$ = $\setof{2}$ in what follows. 
By the admissibility of the weakening rules, it suffices to derive
\[
\setof{1,2}: \neg \theta \to \Exi{x} \varphi(x),  \setof{1}: \neg \theta, \setof{2}: \neg \theta \Rightarrow \setof{1,2}: \Exi{x} (\neg \theta \to \varphi(x)).
\]
We proceed as follows:
\begin{footnotesize}
\[
\infer[(Cut)]{
\setof{1,2}: \neg \theta \to \Exi{x} \varphi(x),  \setof{1}: \neg \theta, \setof{2}: \neg \theta \Rightarrow \setof{1,2}: \Exi{x} (\neg \theta \to \varphi(x))
}{
\setof{1}: \neg \theta, \setof{2}: \neg \theta \Rightarrow \setof{1,2}:\neg \theta 
&
\setof{1,2}:\neg \theta, \setof{1,2}: \neg \theta \to \Exi{x} \varphi(x) \Rightarrow \setof{1,2}: \Exi{x} (\neg \theta \to \varphi(x))
}
\]
\end{footnotesize}
\noindent The topleft sequent is derivable by the rule $(\Rightarrow \neg )$ of Corollary \ref{cor:neg_seqcalc}. 
The topright sequent is derivable similarly to (i). 
This finishes establishing (ii). 
\qed\end{proof}

%\section{Proof of Corollary \ref{cor:neg_seqcalc}}\label{sec:negseqcalc}
%%%%%%%%%%%%%%%%%%%%%%%%%%%%%%%%%%%%%%%%%%%%%%%%%%%%%%%%%%%%%
%%%%%%%%%%%%%%%%%%%%%%%%%%%%%%%%%%%%%%%%%%%%%%%%%%%%%%%%%%%

%\section{Proof of Proposition \ref{prop:kurocalc}} \label{sec:kurocalc}

%%%%%%%%%%%%%%%%%%%%%%%%%%%%%%%%%%%%%%%%%%%%%%%%%%%%%%%%%%%%%%%%
%%%%%%%%%%%%%%%%%%%%%%%%%%%%%%%%%%%%%%%%%%%%%%%%%%%%%%%%%%%%%%%

%\section{Proof of Proposition \ref{prop:KP}} \label{sec:propKP}

%%%%%%%%%%%%%%%%%%%%%%%%%%%%%%%%%%%%%%%%%%%%%%%%%%%%%%%%%%%%%%%%%
%%%%%%%%%%%%%%%%%%%%%%%%%%%%%%%%%%%%%%%%%%%%%%%%%%%%%%%%%%%%%%
%\section{Proof of Proposition \ref{prop:EKP}} \label{sec:proofEKP}

\else\fi

\end{document}